\documentclass[12pt, draftclsnofoot, onecolumn]{IEEEtran}
\usepackage{etex}
\let\tnote\relax
\usepackage{algpseudocode}
\usepackage{algorithm}
\usepackage{gensymb}
\usepackage{amsthm}
\usepackage{threeparttable}

\usepackage{multicol}
\usepackage{empheq}

\newcommand{\Erf}[1]{{\mathbb{E}_{\mathbf{H}_\text{R}}}\left[{#1}\right]}
\newcommand{\Emm}[1]{{\mathbb{E}_{\mathbf{H}_\text{m}}}\left[{#1}\right]}

\newcommand{\vc}[1]{{\mathbf{#1}}}
\newcommand{\comment}[1]{}

\usepackage{mathtools}

\newcommand{\dsp}{\displaystyle}

\usepackage{bm}
\usepackage{amssymb}

\usepackage{varwidth}
\DeclareMathOperator*{\argmax}{arg\,max}
\usepackage[noadjust]{cite}

\usepackage{enumerate}
\usepackage{amssymb}
\usepackage{multirow}
\usepackage{dsfont}
\usepackage{citesort}
\usepackage{array}

\newtheorem{lemma}{Lemma}
\newtheorem{theorem}{Theorem}

\usepackage{listings}
\usepackage[utf8x]{inputenc} 
\usepackage{graphicx}
\usepackage{float}
\usepackage{epstopdf}
\usepackage{bbm}
\usepackage{pbox}
\usepackage{color}
\usepackage[section,subsection,subsubsection]{extraplaceins}
\usepackage{multirow}
\usepackage{subfigure}
\usepackage{booktabs}
\newcommand{\ra}[1]{\renewcommand{\arraystretch}{#1}}
\newcommand{\RNum}[1]{\uppercase\expandafter{\romannumeral #1\relax}}
\newcommand{\suchthat}{\;\ifnum\currentgrouptype=16 \middle\fi|\;}

\usepackage{tikz}
\usetikzlibrary{automata,arrows,positioning,calc}

\usetikzlibrary{arrows}
\tikzstyle{int}=[draw, fill=blue!20, minimum size=2em]
\tikzstyle{init} = [pin edge={to-,thin,black}]

\usetikzlibrary{positioning,chains,fit,shapes,calc, arrows, shadows,trees}
\usepackage{calc}
\usepackage{multicol}
\usepackage{lipsum}
\usetikzlibrary{decorations.markings, arrows, automata}
\tikzstyle{vertex}=[circle, draw, inner sep=0pt, minimum size=5pt]
\tikzset{symbol/.style={rectangle, draw, very thick,
minimum size=10mm, rounded corners=1mm}}
\tikzset{symbol2/.style={rectangle , draw,  thick,
minimum size=35mm, rounded corners=1mm}}

\tikzstyle{block} = [draw, fill=white, rectangle, 
    minimum height=3em, minimum width=6em]
\tikzstyle{sum} = [draw, fill=gray, circle, node distance=1cm]
\tikzstyle{input} = [coordinate]
\tikzstyle{output} = [coordinate]
\tikzstyle{pinstyle} = [pin edge={to-,thick,black}]
\usepackage{mathtools}

\begin{document}
\title{\centering \LARGE{Energy-Efficient Power and Bandwidth Allocation in an Integrated Sub-6 GHz -- Millimeter~Wave System }}
{\author{\IEEEauthorblockN{Morteza Hashemi\IEEEauthorrefmark{1}, C. Emre Koksal\IEEEauthorrefmark{1}, and Ness B. Shroff \IEEEauthorrefmark{1}\IEEEauthorrefmark{2}}\\\IEEEauthorblockA{\IEEEauthorrefmark{1}  Department of Electrical and Computer Engineering, The Ohio State University} \\
\IEEEauthorblockA{\IEEEauthorrefmark{2}Department of Computer Science and Engineering, The Ohio State University}}}
\maketitle
\thispagestyle{plain}
\pagestyle{plain}
\begin{abstract}
In mobile millimeter wave (mmWave) systems, energy is a scarce resource due to the large losses in the channel and high energy usage by analog-to-digital converters (ADC), which scales with bandwidth. In this paper, we consider a communication architecture that integrates the sub-6 GHz and mmWave technologies in $5$G cellular systems. In order to mitigate the energy scarcity in mmWave systems, we investigate the rate-optimal and energy-efficient physical layer resource allocation jointly across the sub-6 GHz and mmWave interfaces. First, we formulate an optimization problem in which the objective is to maximize the achievable sum rate under power constraints at the transmitter and receiver. Our formulation explicitly takes into account the energy consumption in integrated-circuit components, and  assigns the optimal power and bandwidth across the interfaces. We consider the settings with no channel state information and partial channel state information at the transmitter and under high and low SNR scenarios. Second, we investigate the energy efficiency (EE) defined as the ratio between the amount of data transmitted and the corresponding incurred cost in terms of power. We use fractional programming and Dinkelbach's algorithm to solve the EE optimization problem. \emph{Our results prove that despite the availability of huge bandwidths at the mmWave interface, it may be optimal (in terms of achievable sum rate and energy efficiency) to utilize it partially.} Moreover, depending on the sub-6 GHz and mmWave channel conditions and total power budget, it may be optimal to activate only one of the interfaces. 
\end{abstract}

\section{Introduction}
The demand for wireless spectrum is projected to continue growing well into the future, and
will only worsen the currently felt spectrum crunch. For instance, the annual data traffic generated by mobile devices is expected to surpass $130$ exabits by 2020 \cite{khan2011mmwave}. To address the problem of spectrum scarcity for cellular communications, it is envisioned that in $5$G cellular systems certain portions of the mmWave band will be used, spanning the spectrum between $30$ GHz to $300$ GHz \cite{rappaport2013millimeter}. However, before mmWave communications can become a reality, it faces significant challenges such as high propagation losses. The propagation loss at the mmWave band is much higher than the sub-6 GHz frequencies due to a variety of factors including atmospheric absorption, basic Friis transmission-effect, and low penetration.  For instance, materials
such as brick can attenuate mmWave signals by as much as $40$ to $80$ dB \cite{zhao201328} and the
human body itself can result in a $20$ to $35$ dB loss \cite{lu2012modeling}.

In order to compensate for high propagation losses, large antenna arrays with high directivity are needed.  In fact, one of the main drivers behind the emergence of mmWave mobile communications
is the emergence of large antenna arrays that can be deployed in relatively small chip areas. In this case, the mean end-to-end channel
gain is amplified by the product of the gains of the transmitter and receiver antennas. These large antenna-arrays, however, cause several other issues such as high energy consumption, mainly because of the analog-to-digital converters (ADCs) and power
amplifiers. For instance, 
consumption in  ADCs is substantial and it can be written as $P^{\text{(ADC)}} = c_\text{ox} W 2^{r_{\text{ADC}}}$, where $W$ is the bandwidth of the mmWave signal, $r_{\text{ADC}}$ is the quantization rate in bits/sample, and constant $c_\text{ox}$ depends on the gate-oxide capacitance of the converter. At a sampling rate of $1.6$ Gsamples/sec, an $8$-bit quantizer consumes $\approx 250$mW of power. During active transmissions, this would constitute up to $50$\% of the overall power consumed for a typical smart phone. 
Table \ref{tab:power-consumpiton} compares the representative values for power consumption in mmWave against that in sub-6 GHz \cite{orhan2015low,phan2010reconfigurable,liang20070}. 
From this table and the fact that ADC power consumption increases proportionally with bandwidth, we conclude that   \emph{the large bandwidths afforded by mmWave channels present an issue for the components due to the need for a proportionally high power, whereas the achievable data rate increases only logarithmically as a function of the bandwidth. } 

\begin{table}[t]\centering
\renewcommand{\arraystretch}{1.3}
\ra{1.3}
\begin{tabular}{@{}p{4.5cm}cc@{}}\toprule
\textbf{Component} &  \textbf{mmWave Technology} & \textbf{RF Technology}  \\ \toprule
 Analog to Digital Converter (ADC)  &   250 mW & 0.14 mW \\ \midrule
Low Noise Amplifier (LNA) & 39 mW & 3-10 mW \\ \midrule

Frequency Mixer  &  16.8 mW & 0.5-8 mW \\ \midrule
Phase Shifter  & 19.8 mW & -- \\ 
 \bottomrule
\end{tabular}
\caption{Representative values for different component power consumption in the RF and mmWave technologies}
  \label{tab:power-consumpiton}
\end{table}

In addition to power consumption, in designing a communication system, one of the main objectives is to maximize the achievable rate (bits/sec). However, there is a law of diminishing returns, when it comes to the achievable rate, with increasing bandwidth. Indeed, for a wideband coherent communications system, the rate of increase in achievable rate varies as $\frac{\text{SNR}}{W}$ as a function of the bandwidth $W$. Therefore, it is often the case that the achievable rate per unit power is a non-increasing function of available bandwidth beyond a threshold. To illustrate this point,  we calculate the achieved bits/sec/watt for mmWave and sub-6 GHz interfaces in Fig.~\ref{fig:shannon_with_components}, incorporating the consumption by the components. In order to plot this figure, we have used a SISO model and typical values for consumption provided in~\cite{rappaport2011state,phan2010reconfigurable,liang20070} for an 8-bit quantizer and power spectral density of white noise taken to be the product of the temperature and the Boltzman constant. The overall transmission power is taken to be $10$ dB higher for sub-6 GHz and the loss in the channel is also taken into account. Note that the span of the bandwidth values is taken to be between $10$ KHz to $10$ MHz for sub-6 GHz and between $0.7$ MHz to $7$ GHz for mmWave. 

\begin{figure}[t]
\centerline{\includegraphics[height=2.5in]{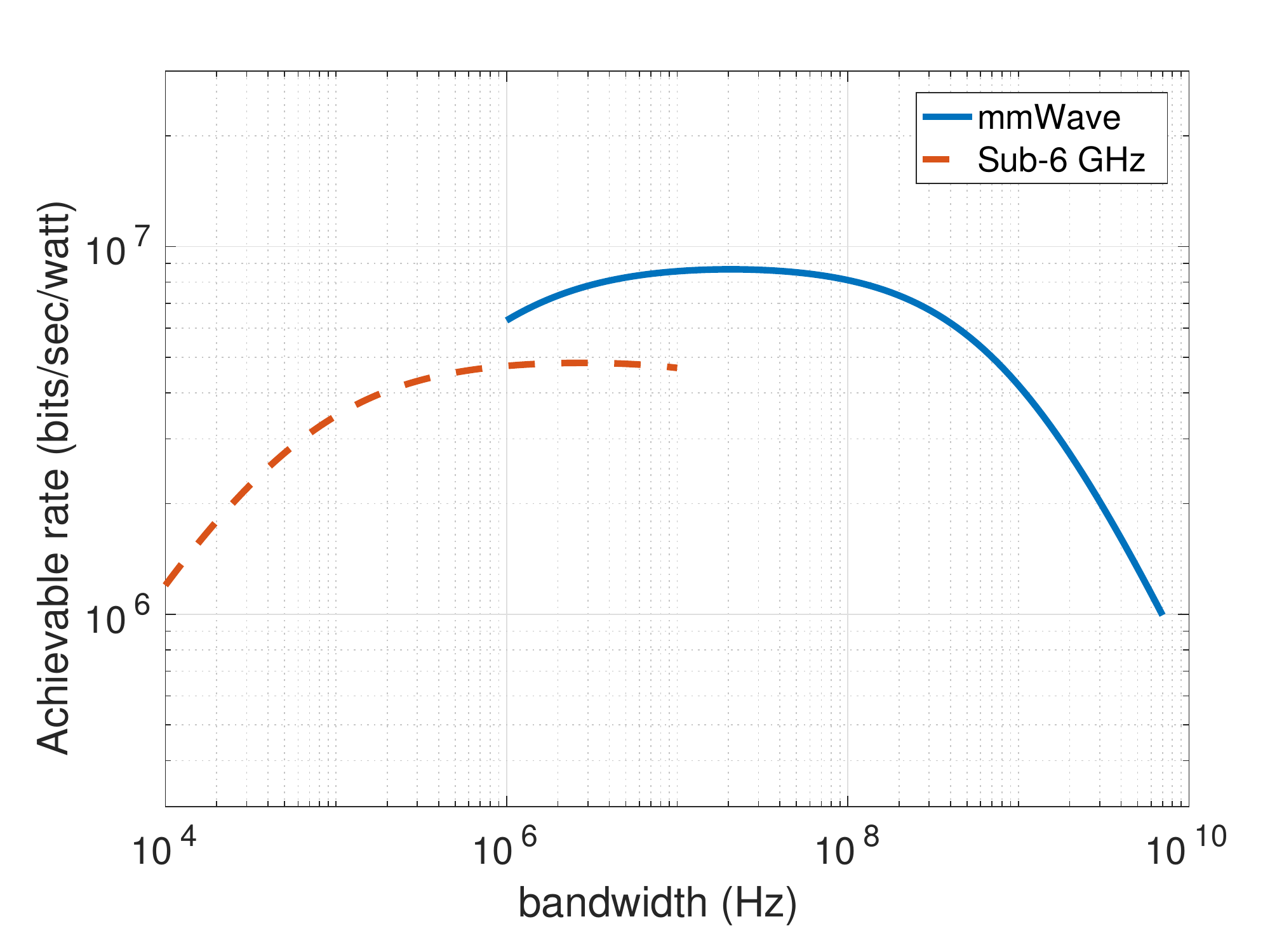}}
\vspace{-0.2in}
\caption{\small{Achievable rate per unit power with the component energy consumption taken into account.}} 
\label{fig:shannon_with_components}
\end{figure}

We have some interesting trends here. Firstly, the  achievable rate per unit bandwidth is not a monotonic function: for mmWave, it tends to decrease for large bandwidth values due to the increased consumption by the ADCs. The amount of increase in rate decreases inversely with the bandwidth, while the ADC power consumption increases linearly with bandwidth. This leads to the reduction in rate per unit power in the wideband regime. On the contrary, in a large band of values, sub-6 GHz interface becomes increasingly energy efficient as the bandwidth increases, due to the relatively low consumption in ADCs and other components. 
Thus, the mmWave interface should be utilized for data communication up to a certain bandwidth (above $\sim 800$MHz in Fig.~\ref{fig:shannon_with_components}). From this example, we conclude that \emph{even though a large bandwidth is available in mmWave band, it may be more energy efficient to utilize only a part of it.} Beyond that point, sub-6~GHz starts to become more energy-efficient per bit transmitted  due to the relatively low consumption in ADCs.

\subsection{Our Contributions}
In order to fully exploit the abundant yet intermittent mmWave bandwidth, we consider an integrated architecture in which the sub-6 GHz and mmWave interfaces coexist and act in cohesion \cite{hashemi2017hybrid}. This paper is aimed to optimally allocate the power and bandwidth jointly across the interfaces. To this end, first we formulate an optimization problem to maximize the achievable sum rate under power constraints at the transmitter and receiver. We investigate the solution space under various conditions depending on the availability of channel state information at the transmitter. Next, we consider the energy efficiency optimization problem, and demonstrate that this problem, in fact, can be reduced to the sum rate maximization. The important point is that our formulation explicitly takes into account the energy consumption in integrated-circuit components. Our results show that despite the availability of huge mmWave bandwidth, it may not be the most energy efficient to utilize it fully at all times. In particular, we prove that:
\begin{enumerate}
\item The ratio of optimal power allocated to the bandwidth utilized for each interface is proportional to the channel state of that interface. Therefore, whenever the channel condition of one of the interfaces degrades, the ratio of power to bandwidth for that interface should decrease as well.

\item As the ADC components become less energy efficient (i.e., a higher power consumption), it is optimal to decrease the utilized mmWave bandwidth. 

\item As the total power budget increases, it is optimal to expand the span of the utilized mmWave bandwidth.  

\item Since the sub-6 GHz bandwidth is much smaller than the mmWave bandwidth by several orders of magnitude, component consumption in sub-6 GHz becomes negligible compared with mmWave. Hence, it is almost optimal to fully utilize the sub-6 GHz bandwidth. 

\item When the system operates at low SNR regime (i.e., low power budget), the input power is allocated to the interface that has a better channel condition. 
\end{enumerate}

In summary, the main contributions of this work are as follows: (i) we provision an integrated sub-6 GHz/mmWave transceiver model, and formulate two optimization problems to  jointly allocate the bandwidth and power across the interfaces in order to maximize the achievable sum rate and energy efficiency; (ii) we investigate the optimal solution space under various conditions such as when there is no channel state information  or when there is only partial channel state information available at the transmitter; and (ii) we analytically confirm the experimental observations in that it may not be optimal (in terms of achievable sum rate) to allocate full bandwidth to both interfaces.

\subsection{Related Work}
We classify existing and related work across the following thrusts: (i) energy efficient mmWave architectures, and (ii) integrated mmWave systems. 

\textbf{(I) Energy efficient mmWave architectures:} Energy efficient transceiver architectures such as the use of low resolution ADCs and hybrid analog/digital combining has attracted significant interest. The limits of communications over additive white Gaussian channel
with low resolution (1-3 bits) ADCs at the receiver is studied
in \cite{singh2009limits}. The bounds on the capacity of the MIMO channel with 1-bit ADC at high and low SNR regimes are derived in \cite{mo2014high} and \cite{mezghani2007ultra},
respectively. The joint optimization
of ADC resolution with the number of antennas in a MIMO channel is studied in \cite{bai2013optimization}. While \cite{el2014spatially} provides efficient hybrid precoding and combining algorithms for sparse mmWave channels that performs close to full digital solution, \cite{alkhateeb2014mimo} combines efficient channel estimation
with the hybrid precoding and combining algorithm in \cite{el2014spatially}. 

Although there has been extensive amount of work to optimize the mmWave receivers architecture (e.g., in terms of ADCs), the effect of bandwidth on the mmWave performance has not been fully investigated.  To the best of our knowledge, only  the authors in \cite{orhan2015low} have studied  the effect of bandwidth on the performance of \emph{standalone mmWave systems}. Compared with \cite{orhan2015low}, we consider an \emph{integrated sub-6 GHz/mmWave architecture} in which the transmitter and receiver are power constrained. In this case,  an optimal power and bandwidth allocation is derived to maximize the achievable sum rate and energy efficiency.  In addition, the authors in \cite{orhan2015low} consider the number of ADC quantization bits as an optimization parameter, while we assume that the ADC architecture is fixed and the transmit power and bandwidth are optimized  for a MIMO architecture. We derive the closed form expressions for the optimal power and bandwidth allocation across the sub-6 GHz and mmWave interfaces. To the best of our knowledge, there is no previous work that investigates the joint effect of bandwidth and transmission power on the performance of integrated sub-6 GHz/mmWave systems.

\textbf{(II) Integrated sub-6 GHz/mmWave systems:}
 Beyond the classical mmWave communications and beamforming methods \cite{rappaport2013millimeter,collonge2004influence,roh2014millimeter,adhikary2014joint}, recently, there have been proposals  on leveraging out-of-band information in order to enhance the mmWave performance. The authors in \cite{aliestimating} propose a transform method to translate the spatial correlation matrix at the sub-6 GHz band into the correlation matrix of the mmWave channel. The authors in \cite{nitsche2015steering} consider the $60$ GHz indoor WiFi network, and investigate the correlation between the estimated angle-of-arrival (AoA) at the sub-6 GHz band with the mmWave AoA in order to reduce the beam-steering
overhead. The authors in \cite{ali2017millimeter} propose a compressed beam selection method which is based on out-of-band spatial information obtained at sub-6 GHz band. Our work is distinguished from the above cited works as we investigate the optimal physical layer resource allocation across the sub-6 GHz and mmWave interfaces.  In our proposed hybrid sub-6 GHz/mmWave architecture, both interfaces can be simultaneously used for data transfer.  In \cite{hashemi2017hybrid},  we investigated the problem of optimal load devision and scheduling across the sub-6 GHz and mmWave interfaces.

\subsection{Notations} We use the following notation throughout the paper. Bold uppercase and lowercase letters are used for matrices and vectors, respectively, while non-bold letters are used for scalers.  In addition, $(.)^H$ denotes the conjugate transpose; $\text{tr}(.)$ denotes the matrix trace operator; and $\mathds{E}[.]$ denotes the expectation operator. The sub-6 GHz and mmWave variables are denoted by $(.)_\text{sub-6}$ and $(.)_\text{m}$, respectively.

\section{System Model and Problem Formulation}
\subsection{System Model}
Figure \ref{fig:system} illustrates the  components of our proposed architecture that integrates the sub-6 GHz and mmWave interfaces. In \cite{hashemi2017hybrid}, we demonstrated that sub-6 GHz spatial information  enables the mmWave beamforming fully in the analog domain without incurring large delay overhead, and thus it can remedy the high energy consumption issue by reducing the number of ADC components that are otherwise needed for digital beamforming. 
 \begin{figure*}[t!]
\begin{center}
\includegraphics[scale=.33, trim = 0cm 3.7cm 0cm 3.2cm, clip]{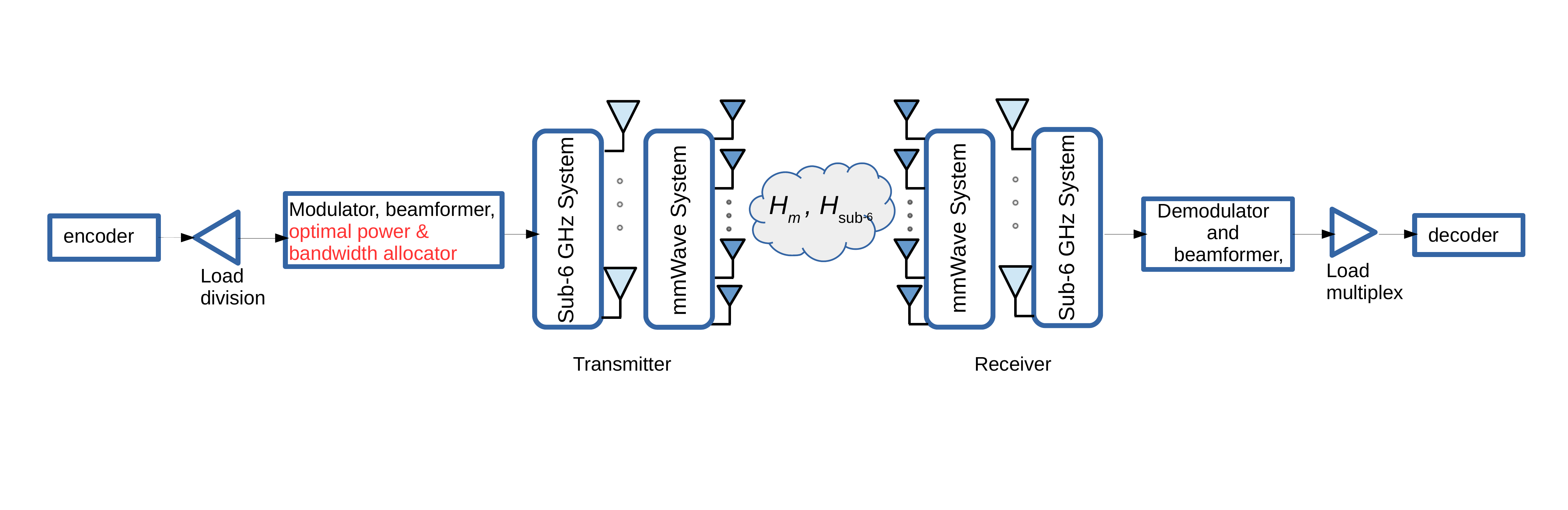}
\caption{Integrated sub-6 GHz and mmWave system with optimal power and bandwidth allocation across the interfaces. The goal is to maximize the achievable sum rate and energy efficiency. }
\label{fig:system}
\end{center}
\end{figure*}
In this paper, in addition to beamforming, we leverage the sub-6 GHz interface for communications and data transfer, and assume that the sub-6 GHz/mmWave transmitter and receiver are power constrained. The power constraint at the transmitter dictates the optimal power allocation across the interfaces, while the receiver power constraint determines the optimal bandwidth allocation. Without loss of generality, we assume that the transmitter and receiver constraints are jointly considered as a single constraint, and the problem is expressed in joint power and bandwidth allocation across the interfaces with the total power budget $P_{\text{max}}$. In this case, the results will qualitatively be parallel to the scenario where we impose constraints on the consumed power at the transmitter and at the receiver separately.


\begin{figure}[t]
\centering
\includegraphics[scale=.4, trim = 0cm 1.5cm 0cm 2.8cm, clip]{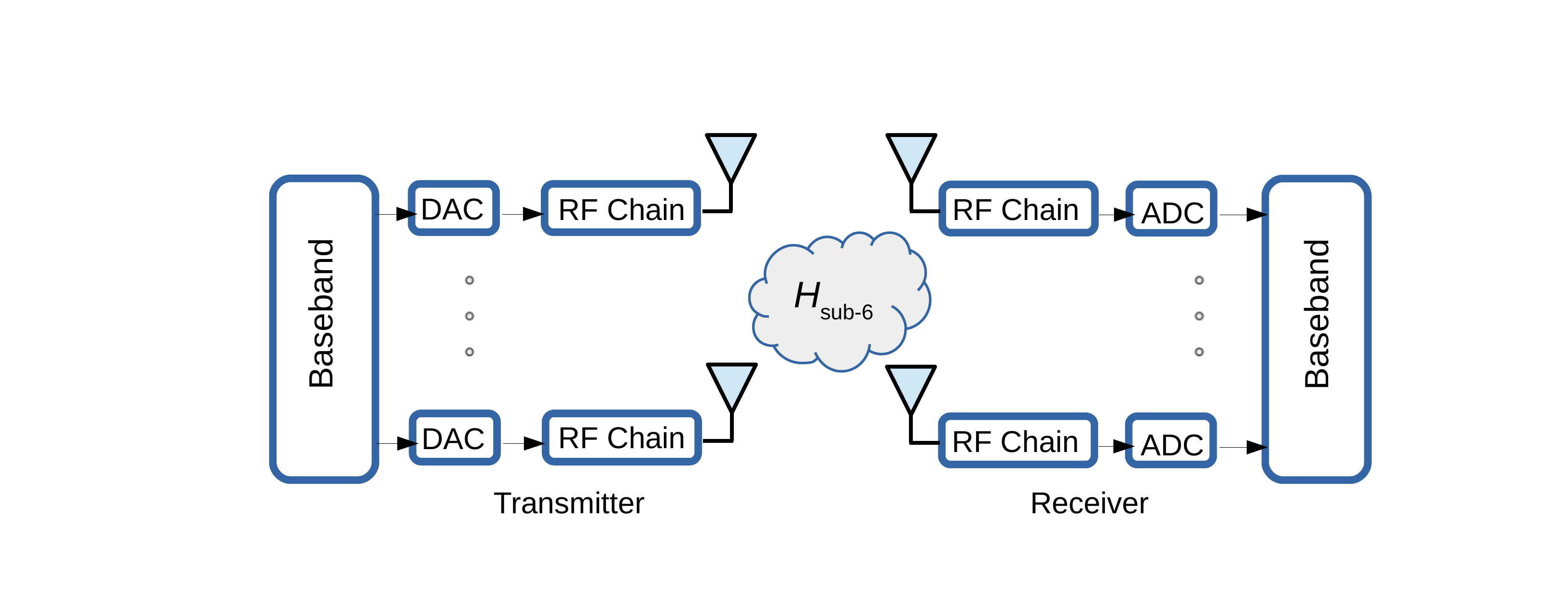}
\caption{Sub-6 GHz system model based on digital beamforming}
\label{fig:RF-model}
\end{figure}

\subsection{Sub-6 GHz System and Channel Model}
The sub-6 GHz system model is shown in Fig. \ref{fig:RF-model} where we use digital beamforming. As a result, the received signal at the receiver can be written as: 
\begin{equation}
\label{eq:rec_signals}
\mathbf{y}_{\text{sub-6}}=\mathbf{H}_{\text{sub-6}} \cdot
\mathbf{x}_{\text{sub-6}}+\mathbf{n}_{\text{sub-6}},
\end{equation}
where $\mathbf{H}_{\text{sub-6}}$ is the sub-6 GHz channel matrix and $\mathbf{x}_{\text{sub-6}}$ is the transmitted signal vector in sub-6 GHz. Entries of circularly symmetric white Gaussian noise $\mathbf{n}_{\text{sub-6}}$ are normalized to have unit variance. 
\begin{figure}[t]
\centering
\includegraphics[scale=.4, trim = 0cm 1.2cm 0cm 0.5cm, clip]{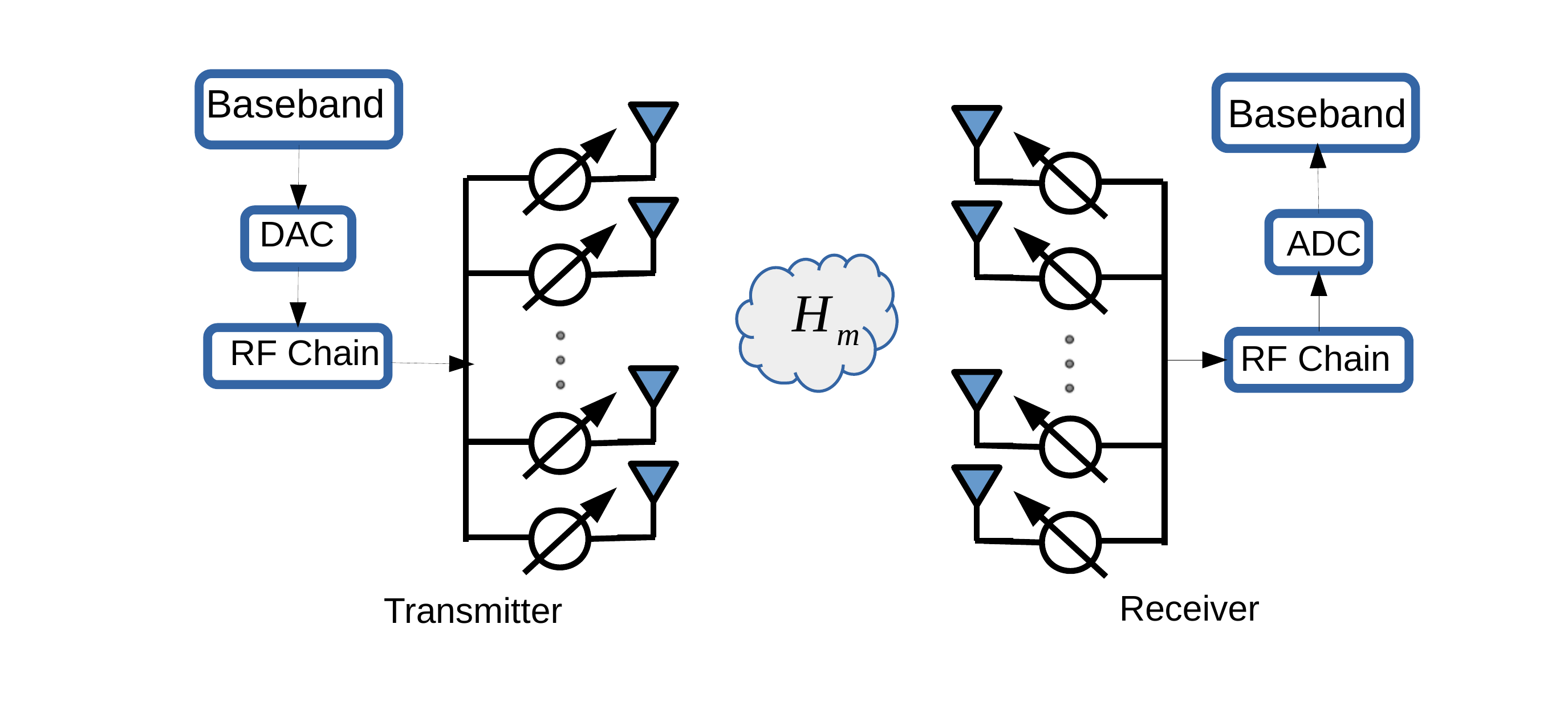}
\caption{mmWave system model with analog beamforming}
\label{fig:mmWave-model}
\end{figure}
In the proposed system, we assume that the sub-6 GHz interface can utilize the total bandwidth of $W_\text{sub-6}^{\max}$. Moreover, the transmission power of the sub-6 GHz interface is denoted by $P_\text{sub-6}=\text{tr}({\mathbf{K}_{\textbf{xx}}})$ in which $\mathbf{K}_{\textbf{xx}}$ is the covariance matrix of signal $\mathbf{x}_\text{sub-6}$.  We assume that the sub-6 GHz system includes $n_t$ transmit and $n_r$ receive antennas.

\subsection{MmWave System and Channel Model} 
 The mmWave system model is shown in Fig. \ref{fig:mmWave-model}. Unlike sub-6 GHz, we use  analog combining for mmWave via a single ADC. Consequently, the signal at the input of the decoder is a scalar, identical to a weighted combination of signal $x_{\text{m}}$ across all antennas. Thus, the received signal at the mmWave receiver can be written as: 
\begin{equation}
\label{eq:rec_signals}
 y_{\text{m}}=\vc{w}_r^H \mathbf{H}_{\text{m}} \vc{w}_t \cdot x_{\text{m}}+\tilde{n}_{\text{m}},
\end{equation}
where $\vc{w}_r$ and $\vc{w}_t$ are the receive and transmit beamforming vectors.  
The $\tilde{n}_{\text{m}}$ term denotes the effective Gaussian noise and is normalized to have unit variance.  The mmWave interface is assigned with the total bandwidth of $W_\text{m}^{\max}$, and mmWave transmit power is denoted by $P_\text{m}$.

\subsection{Problem Formulation}
In mmWave communications, due to the large consumption of the components and the losses in the channel, energy is
a scarce resource, and thus we assume that $P_{\max}$ is the maximum power available for \emph{data transmission} and \emph{component consumption} (i.e., ADC components) across the sub-6 GHz and mmWave interfaces. Moreover, the ADC power consumption scales proportionally with bandwidth, i.e., $P^{\text{(ADC)}} = a W$, where $W$ denotes the bandwidth and $a$ is a constant for a given ADC with fixed quantization rate (i.e., $a = c_{ox} 2^{r_{ADC}}$). As mentioned earlier and from Fig. \ref{fig:shannon_with_components}, it may not be the most energy efficient to utilize the available bandwidth fully at all times, and thus bandwidth allocated to each interface needs to be optimized in addition to the power allocated to each interface.  Hence, we define Problem 1 as follows in order to maximize the achievable sum rate across the interfaces with a given constraint on the joint transmitter and receiver power consumption.  
\begin{itemize}
\item \textbf{Problem 1} \textbf{(Power-Constrained Sum Rate Maximization):} \emph{We consider the problem of maximizing the sum rate of the sub-6 GHz and mmWave interfaces subject to a power constraint, i.e.,:}
\begin{subequations}
\small
\begin{align}
& \dsp{\max_{\substack{W_{\text{m}},W_{\text{sub-6}} \\ P_{\text{m}}, P_{\text{sub-6}}}}}  \dsp{\Erf{ W_{\text{sub-6}} \log \det \left( \mathbf{I}+ \frac{\mathbf{H}_{\text{sub-6}}^H \mathbf{K}_{\mathbf{xx}} \mathbf{H}_{\text{sub-6}}}{W_{\text{sub-6}}}  \right)}}  \dsp{+\Emm{W_{\text{m}} \log \left( 1 + \frac{P_{\text{m}} \left|\vc{w}_r^H \mathbf{H}_{\text{m}} \vc{w}_t\right|^2}{W_{\text{m}}} \right)}} \label{problem-1} \\
 & \ \text{subject to:}  \ \ \  P_{\text{sub-6}} \large{\mathds{1}}_{W_\text{sub-6} > 0} + n_r aW_{\text{sub-6}}+P_{\text{m}} \large{\mathds{1}}_{W_\text{m} > 0}+ a W_{\text{m}}  \leq P_{\max},  \label{power-constraint}\\
& \hspace{2.2cm} 0 \leq W_\text{sub-6} \leq W_\text{sub-6}^{\max}; \  0 \leq W_{\text{m}} \leq W_\text{m}^{\max}; \label{bandwidth-constraint}\\ & \hspace{2.2cm} 0 \leq P_\text{m}, P_\text{sub-6}.  
\end{align}
\end{subequations} 
\normalsize

In Section \ref{sec:problem1-solution}, we investigate the optimal solution of Problem 1.
\end{itemize}

\begin{itemize}
\item \textbf{Problem 2 (Energy Efficiency Maximization):}
\emph{The second problem that we explore is maximization of energy efficiency (sum rate per unit power expenditure), i.e.,:}
\begin{subequations}
\small
\begin{align}
& \dsp{\max_{\substack{W_{\text{m}},W_{\text{sub-6}} \\ P_{\text{m}}, P_{\text{sub-6}}}}}  \frac{\dsp{\Erf{W_{\text{sub-6}} \log \det \left( \mathbf{I}+ \frac{\mathbf{H}_{\text{sub-6}}^H \mathbf{K}_{\mathbf{xx}} \mathbf{H}_{\text{sub-6}}}{W_{\text{sub-6}}}\right)} + \Emm{W_{\text{m}} \log  \left( 1+ \frac{P_{\text{m}}}{W_{\text{m}}}\left|\vc{w}_r^H \mathbf{H}_{\text{m}} \vc{w}_t\right|^2 \right)}}}{P_{\text{sub-6}} \large{\mathds{1}}_{W_\text{sub-6} > 0} + n_r aW_{\text{sub-6}} + P_{\text{m}} \large{\mathds{1}}_{W_\text{m} > 0}+ a W_{\text{m}}}, \label{eq:problem2} \\
& \ \text{subject to:}  \ \ 0 \leq W_\text{sub-6} \leq W_\text{sub-6}^{\max}; \  0 \leq W_{\text{m}} \leq W_\text{m}^{\max};  \label{eq:problem2-bandwidth-constraint}\\ & \hspace{2cm} 0 \leq P_\text{m}, P_\text{sub-6}.  \label{eq:problem2-power-constraint}
 \end{align}
 \end{subequations}
 \normalsize
 In Section \ref{sec:problem2-solution}, we consider the solution space of Problem 2. 
\end{itemize}

\section{Power-Constrained Sum Rate Maximization}
\label{sec:problem1-solution}
 In this section, we investigate the optimal solution of Problem 1. To this end, we note that Problem 1 is that of the convex optimization since the objective function is concave and the constraint is linear. In addition, the objective function is concave and increasing in the variables $W_{\text{sub-6}}, P_\text{sub-6}, W_{\text{m}}$, and $P_\text{m}$. It is straightforward to show that the objective function is increasing in $P_\text{m}$ and $P_\text{sub-6}$. Moreover, by taking the derivative of the function $f(x) = x\log(1+\frac{1}{x})$, one can see that $f(x)$ is increasing in $x$, and thus \eqref{problem-1} is increasing in $W_\text{sub-6}$ and $W_\text{m}$ as well. Using the second order derivatives, we can show that the objective function is concave in $W_\text{sub-6}$ and $W_\text{m}$. 
 
 From Problem 1, we note that there is a tradeoff in bandwidth allocation: it is desirable to set the sub-6 GHz and mmWave bandwidth variables to $W_\text{sub-6}^{\max}$ and $W_\text{m}^{\max}$ in order to increase the objective value. However, due to the \eqref{power-constraint} constraint, high bandwidth reduces the \emph{transmission power} (due to the power-hungry ADC components), which in turn, reduces the objective value. Similarly, there is a tradeoff in allocating the  transmission power $P_\text{sub-6}$ and $P_\text{m}$.
In order to optimally balance this tradeoff, we solve the sum rate optimization problem using the convex optimization tools under two scenarios: (i) when no channel state information at the transmitter (CSIT) is available, and (ii) when partial CSIT is available.  
 
\subsection{No Channel State Information at the Transmitter}
When the channel matrix $\mathbf{H}_\text{sub-6}$ is random and there is no channel state information at the transmitter (CSIT) available, the optimal power allocation across the $n_t$ antenna elments of the sub-6 GHz interface (different than the optimal power allocation across the interfaces), is uniform \cite{tse2005fundamentals}. Therefore, the covariance matrix $\mathbf{K}_\mathbf{xx}$ can be written as:
\begin{equation}
\mathbf{K}_\mathbf{xx} = \frac{P_\text{sub-6}}{n_t} \mathbf{I}_{n_t}, 
\label{eq:optimal-covariance-no-csi}
\end{equation}
\noindent and thus, the first term in \eqref{problem-1} is simplified to:
\begin{equation}
\dsp{W_{\text{sub-6}} {\log \det \left( \mathbf{I}+ \frac{P_\text{sub-6}}{W_{\text{sub-6}} n_t} \mathbf{H}_{\text{sub-6}}^H \mathbf{H}_{\text{sub-6}} \right)}}.
\label{rf-capacity}
\end{equation}
We assume that $\lambda_1 \geq \lambda_2 \geq ... \geq \lambda_{n}$ denote the ordered singular values of the sub-6 GHz channel matrix $\mathbf{H}_{\text{sub-6}}$ where $n = \min(n_t, n_r)$. Therefore, from the determinant and singular value properties, we can rewrite \eqref{rf-capacity} as:
\begin{equation}
W_{\text{sub-6}}\sum_{i=1}^{n} \log \left(1+\frac{P_\text{sub-6}}{W_\text{sub-6} n_t} \lambda_i ^2 \right).
\label{eq:RF-simple-form}
\end{equation}
As a result, \eqref{problem-1} is simplified by replacing its first term with \eqref{eq:RF-simple-form}. Due to the fact that problem \eqref{problem-1} is convex, the
Karush–Kuhn–Tucker (KKT) conditions are necessary and sufficient for the optimality of the solution \cite{boyd2004convex}. In order to derive the KKT conditions, we form the following Lagrangian function:

\vspace{-.6cm}
\small{
\begin{align}
\mathcal{L} (W_\text{m}, W_\text{sub-6}, P_\text{m}, P_\text{sub-6}, \boldsymbol{\mu}) = W_{\text{sub-6}}\dsp{\sum_{i=1}^{n}  \log \left(1+\frac{P_\text{sub-6}}{W_\text{sub-6} n_t} \lambda_i ^2 \right) } 
\dsp{+{W_{\text{m}} \log \left( 1 + \frac{P_{\text{m}}}{W_{\text{m}}} \left|\vc{w}_r^H  \mathbf{H}_{\text{m}}   \vc{w}_t\right|^2 \right)}} \nonumber \\ 
+ \mu_0\bigg(P_\text{max} - \dsp{P_{\text{sub-6}} - n_r a W_{\text{sub-6}}-P_{\text{m}}- a W_{\text{m}}}\bigg) + \mu_1 \left(W_\text{sub-6}^{\max} - W_\text{sub-6}\right) + \mu_2 (W_\text{m}^{\max} - W_\text{m}) + \mu_3 P_\text{sub-6} + \mu_4 P_\text{m} & ,
\label{lagrange}
\end{align}}
\normalsize

\vspace{-.5cm}
\noindent for the Lagrange multiplier vector $\boldsymbol{\mu} = (\mu_0, ..., \mu_4)$. From the KKT stability conditions, we have:
\begin{equation}
  \left\{
  \begin{array}{l l}
n \log \frac{P_\text{sub-6}}{W_\text{sub-6} n_\text{t}} + \sum_{i=1}^{n}  
\log \lambda_i^2 - \frac{P_\text{sub-6} \lambda_i ^2}{W_\text{sub-6} n_t + P_\text{sub-6} \lambda_i^2}  & = n_r a \mu_0 + \mu_1,\\ \vspace{.1cm}
    
W_\text{sub-6} \sum_{i=1}^{n} \frac{\lambda_i^2}{W_\text{sub-6} n_t + P_\text{sub-6} \lambda_i^2} & = \mu_0 - \mu_3, \\ \vspace{.1cm}  
  
  \log \frac{P_\text{m}}{W_\text{m}} + \log A - \frac{P_\text{m}A}{W_\text{m} + P_\text{m} A} &= a  \mu_0 + \mu_2, \vspace{.14cm} \\  
    
 \frac{W_\text{m} A }{W_\text{m} + P_\text{m} A}  & = \mu_0 - \mu_4.
%
  \end{array} \right.
  \label{kkt}
\end{equation}
For the sake of notations, we define $A~:= ~\left|\vc{w}_r^H \mathbf{H}_{\text{m}} \vc{w}_t\right|^2$. 
From the KKT complimentary slackness conditions, we have:
\begin{equation}
  \left\{
  \begin{array}{l l}
  \mu_0\bigg(P_\text{max} - \dsp{P_{\text{sub-6}} - n_r aW_{\text{sub-6}}-P_{\text{m}}- a_{\text{m}}W_{\text{m}}}\bigg) & = 0, \vspace{.14cm} \\  

\mu_1 \left(W_\text{sub-6}^{\max} - W_\text{sub-6}\right)  & = 0, \vspace{.14cm} \\   

\mu_2 (W_\text{m}^{\max} - W_\text{m}) & = 0, \vspace{.14cm} \\ 

\mu_3 P_\text{sub-6} & = 0, \vspace{.14cm} \\ 
   
\mu_4 P_\text{m} & = 0. 
\end{array} \right.
  \label{kkt-slackness}
\end{equation}
 From the second and fourth equations in \eqref{kkt}, we conclude that $\mu_0 - \mu_3 > 0 \Rightarrow \mu_0 > \mu_3 \geq 0$, and thus $\mu_0 > 0$. However, from the first equation in \eqref{kkt-slackness}, we conclude that  $$P_\text{max} - \dsp{P_{\text{sub-6}} - n_r aW_{\text{sub-6}}-P_{\text{m}}- a_{\text{m}}W_{\text{m}}} = 0,$$ and thus the power constraint is satisfied with equality. 
%
Intuitively, we can also argue that the objective function in Problem 1 is an increasing function of the bandwidth parameters $W_\text{sub-6}$ and $W_\text{m}$, and thus the objective function can be  improved by increasing the value of $W_\text{sub-6}$ and $W_\text{m}$. In the case that constraint \eqref{bandwidth-constraint} becomes passive, the objective function can be improved by increasing the transmit power $P_\text{sub-6}$ and $P_\text{m}$ since the objective function is increasing in terms of these variables as well. Therefore, the power constraint \eqref{power-constraint} is always satisfied with equality at the optimal solution. Next, in order to characterize the optimal solution we note that since the transmit power $P_\text{sub-6}$ and $P_\text{m}$ cannot be zero, $\mu_3 = \mu_4 =0$ should hold to satisfy the complimentary slackness conditions. For the Lagrange multipliers $\mu_1$ and $\mu_2$, we have the following cases.

\textbf{(I) Full Sub-6 GHz Bandwidth Allocation ($\boldsymbol{\mu_1 >0}$):} Depending on the channel conditions, if $\mu_1 > \mu_2$ holds, since the Lagrange multiplier $\mu_2$ is nonnegative, it results in $\mu_1>0$. Thus, from the complimentary slackness condition, we conclude that
$W_\text{sub-6}^* = W_\text{sub-6}^{\max}$. Moreover, the condition $\mu_1 > \mu_2$ implies that $n_r a \mu_0 + \mu_1 > a \mu_0 + \mu_2$ as well. As a result, from \eqref{kkt}, we can see that if the sub-6 GHz interface has a better channel condition than the mmWave channel, then the whole available bandwidth $W_\text{sub-6}^{\max}$ should be utilized.  This results in a system of equations from which the optimal value for $W_\text{m}^*, P_\text{m}^*,$ and $P_\text{sub-6}^*$ are calculated as follows: 
\begin{equation}
P_\text{sub-6}^* =   \max\left(0, \frac{n n_r a W_\text{sub-6}^\text{max}}{B}\right),
\label{eq:power-optimal}
\end{equation}
\begin{equation}
P_\text{m}^* =  \max\left(0, \frac{P_\text{max} - CW_\text{sub-6}^\text{max}}{B+1}\right),
\end{equation}
\begin{equation}
W_\text{m}^* =  \max\left(0, \frac{P_\text{max} - C W_\text{sub-6}^\text{max}}{a+\frac{a}{B}}\right),
\label{eq:BW-optimal}
\end{equation}
where:
$$
B = \omega \left(\log(a A) -1 \right) \ \text{and} \ \ C = n_r a + \frac{n n_r a}{B}, 
$$
in which $\omega(.)$ is the Wright omega function. 
Note that in order to make the calculation more tractable, we assume that the system operates at high SNR regime, and thus the following approximations hold:
$$
1+\frac{P_\text{sub-6}}{W_\text{sub-6} n_t} \lambda_i ^2  \approx  \frac{P_\text{sub-6}}{W_\text{sub-6} n_t} \lambda_i ^2,
\quad \text{and} \quad
1 + \frac{P_{\text{m}}}{W_{\text{m}}} A \approx \frac{P_{\text{m}}}{W_{\text{m}}} A. 
$$

Moreover, the KKT conditions in \eqref{kkt} result in the following equations:
\begin{equation}
 \frac{A }{1 + \frac{P_\text{m}}{W_\text{m}} A} = \sum_{i=1}^{n}  \frac{\frac{\lambda_i^2}{n_t}}{1 + \frac{P_\text{sub-6}}{W_\text{sub-6} n_t} \lambda_i^2};
\label{eq:kkt-result-1}
\end{equation}
\begin{equation}
\frac{\log \frac{P_\text{m}}{W_\text{m}} + \log A}{\frac{P_\text{m}}{W_\text{m}} + a} = \frac{\log \frac{P_\text{sub-6}}{W_\text{sub-6}} + \sum_{i=1}^{n}\log \lambda_i^2}{\frac{P_\text{sub-6}}{W_\text{sub-6}} + n_r a},
\label{eq:kkt-result-2}
\end{equation}
in which $A = ~\left|\vc{w}_r^H \mathbf{H}_{\text{m}} \vc{w}_t\right|^2$ captures the mmWave channel conditions. 
\emph{From \eqref{eq:kkt-result-1} and \eqref{eq:kkt-result-2}, we observe that whenever the channel condition of one of the interfaces degrades, the ratio of power to bandwidth for that interface should decrease as well. }

\begin{figure}[t!]
  \begin{center}
    \subfigure[MmWave bandwidth vs. ADC consumption] 
    {
      \label{fig:mmWave_bandwidh_ADC}
      \includegraphics[scale=.21, trim = 1cm 0cm 3cm 1cm, clip]{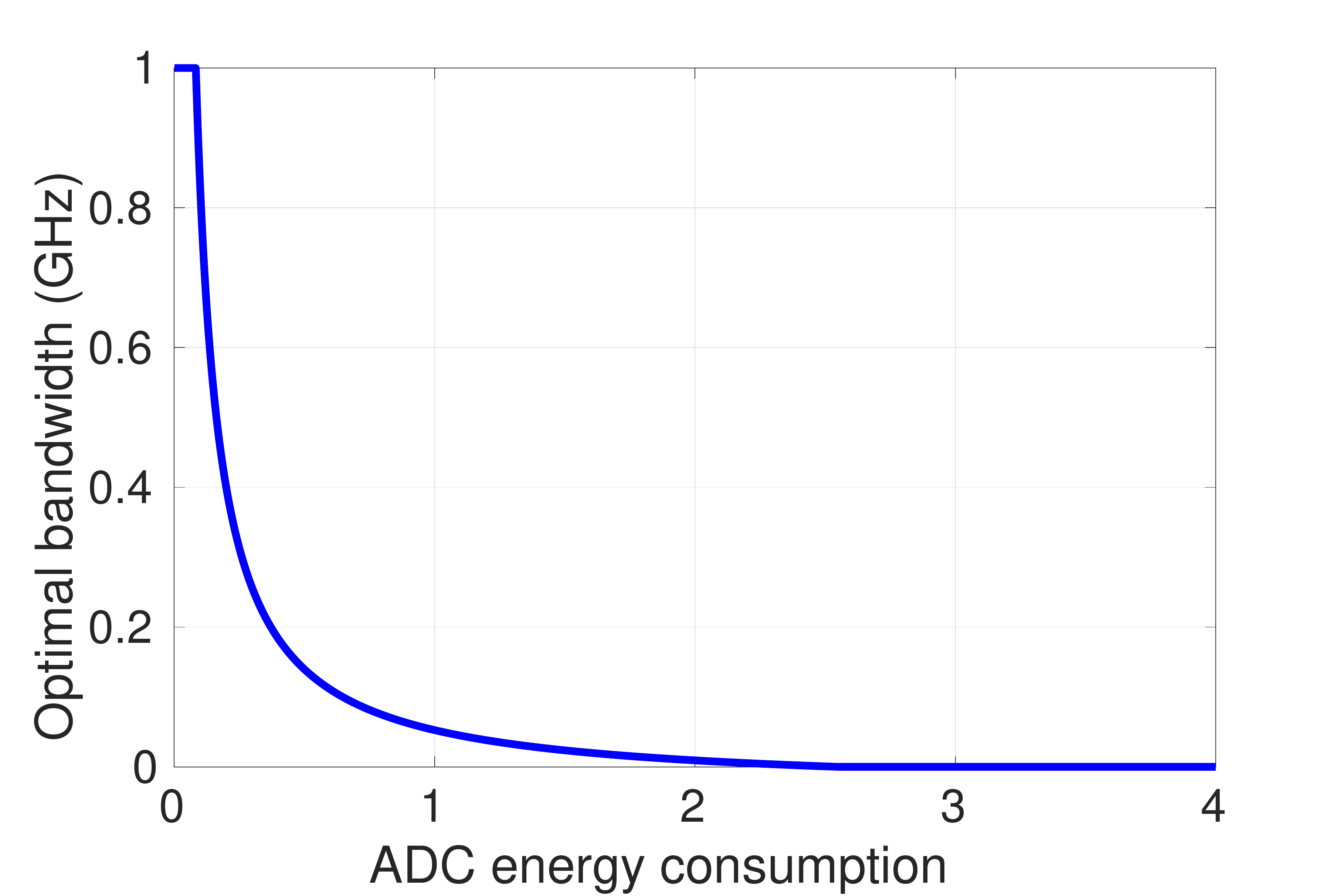}  
    }\hspace{1cm}
     \subfigure[Sum rate vs. ADC consumption] 
    {
      \label{fig:sum_rate_ADC}
      \includegraphics[scale=.21, trim = 0cm 0cm 3cm 1cm, clip]{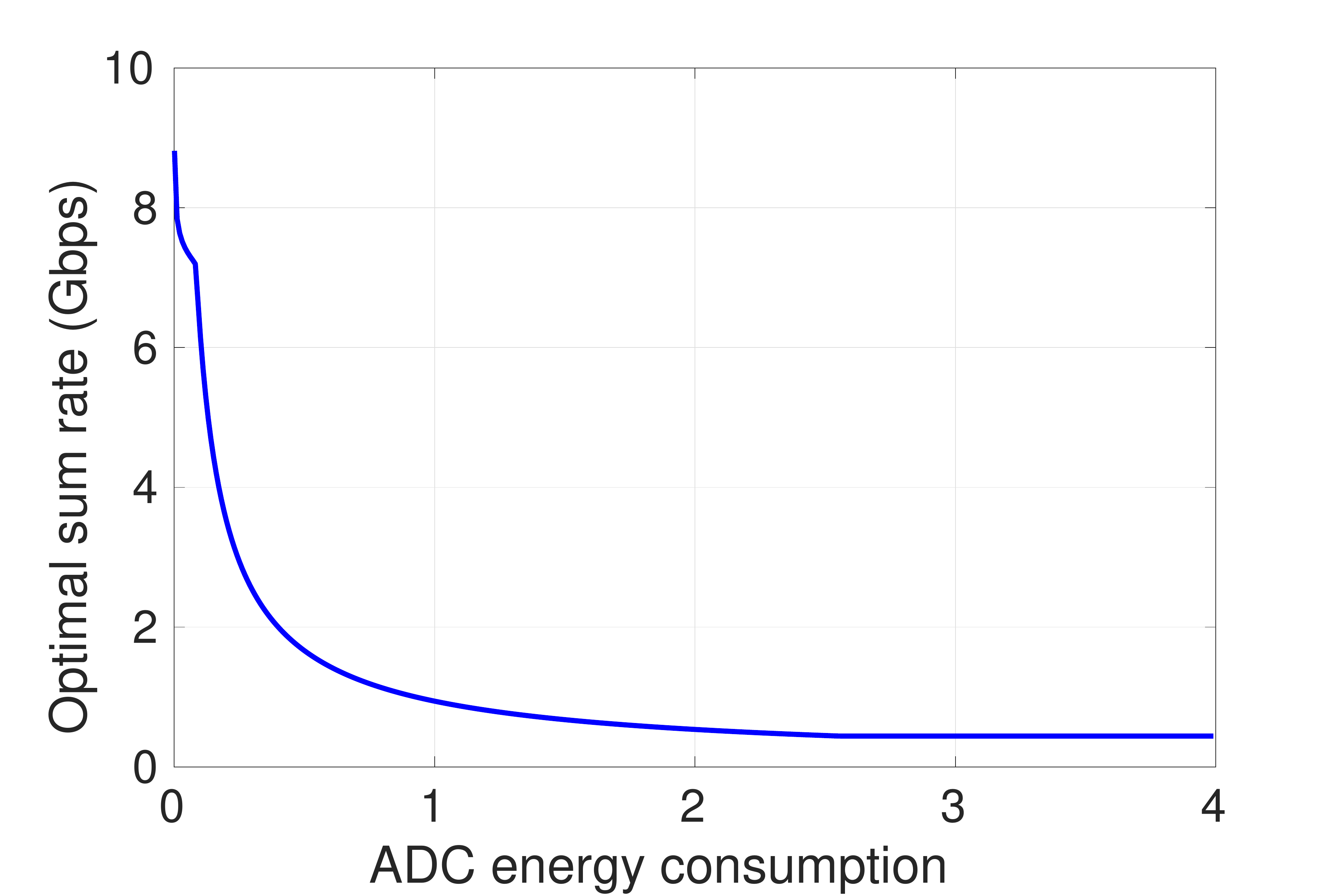}  
    }
  \end{center}
  \vspace{-0.15in}
  \caption{Optimal solution as a function of ADC energy consumption for $P_{\max} = 1000$ mW. In these figures, ADC energy consumption refers to the constant $a \times 10^{8}$. }
  \vspace{-0.2in}
  \label{HB_effect}
\end{figure}

\begin{figure}[t!]
  \begin{center}
    \subfigure[MmWave bandwidth vs. input power] {
      \label{fig:mmWave_bandwidth_power}
      \includegraphics[scale=.28, trim = 3cm 2.8cm 3cm 2.5cm, clip]{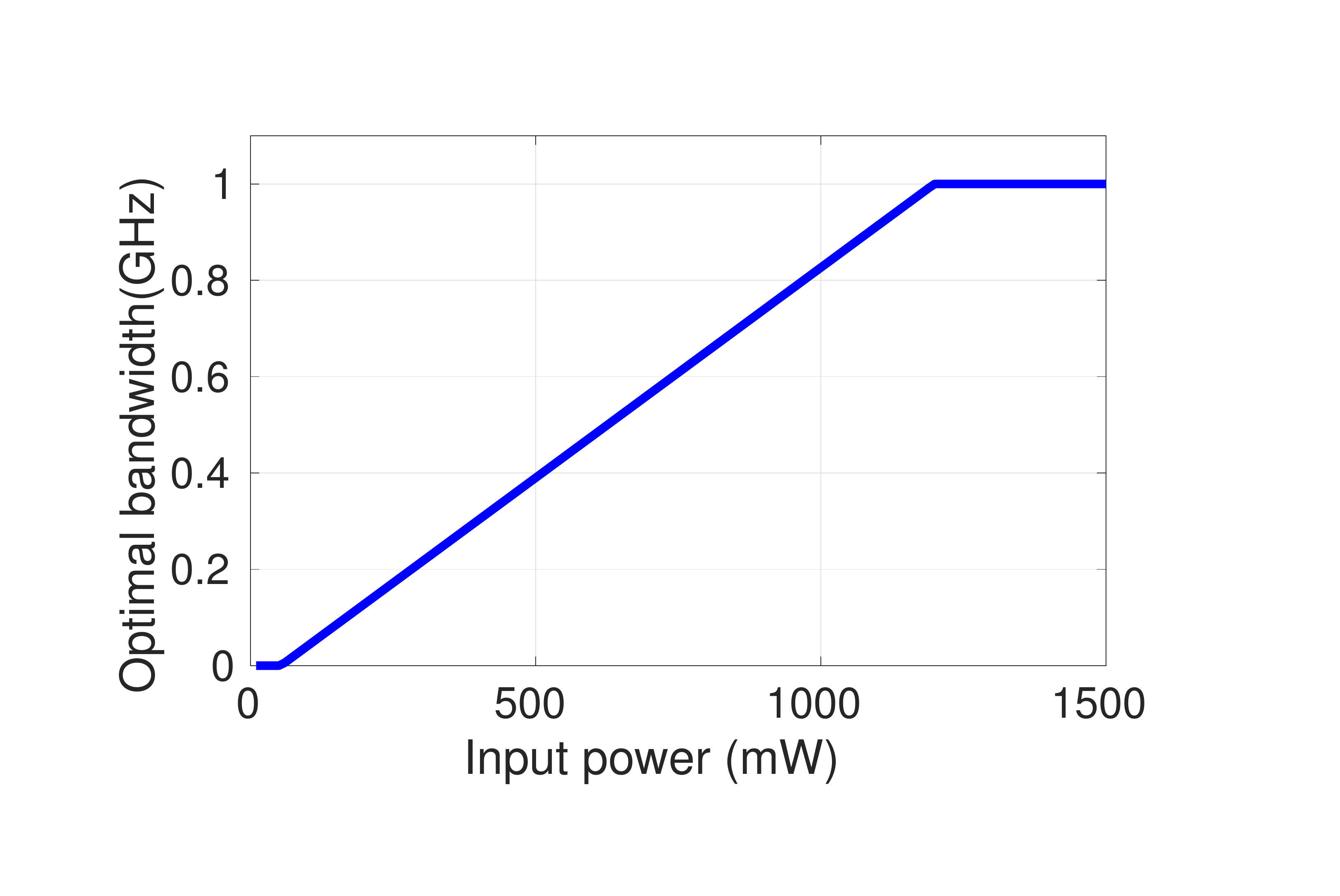} 
    } \hspace{.2cm}
    \subfigure[Sum rate vs. input power] 
    {
      \label{fig:mmWave_sum_rate}
      \includegraphics[scale=.28, trim = 3cm 2.3cm 3cm 2.5cm, clip]{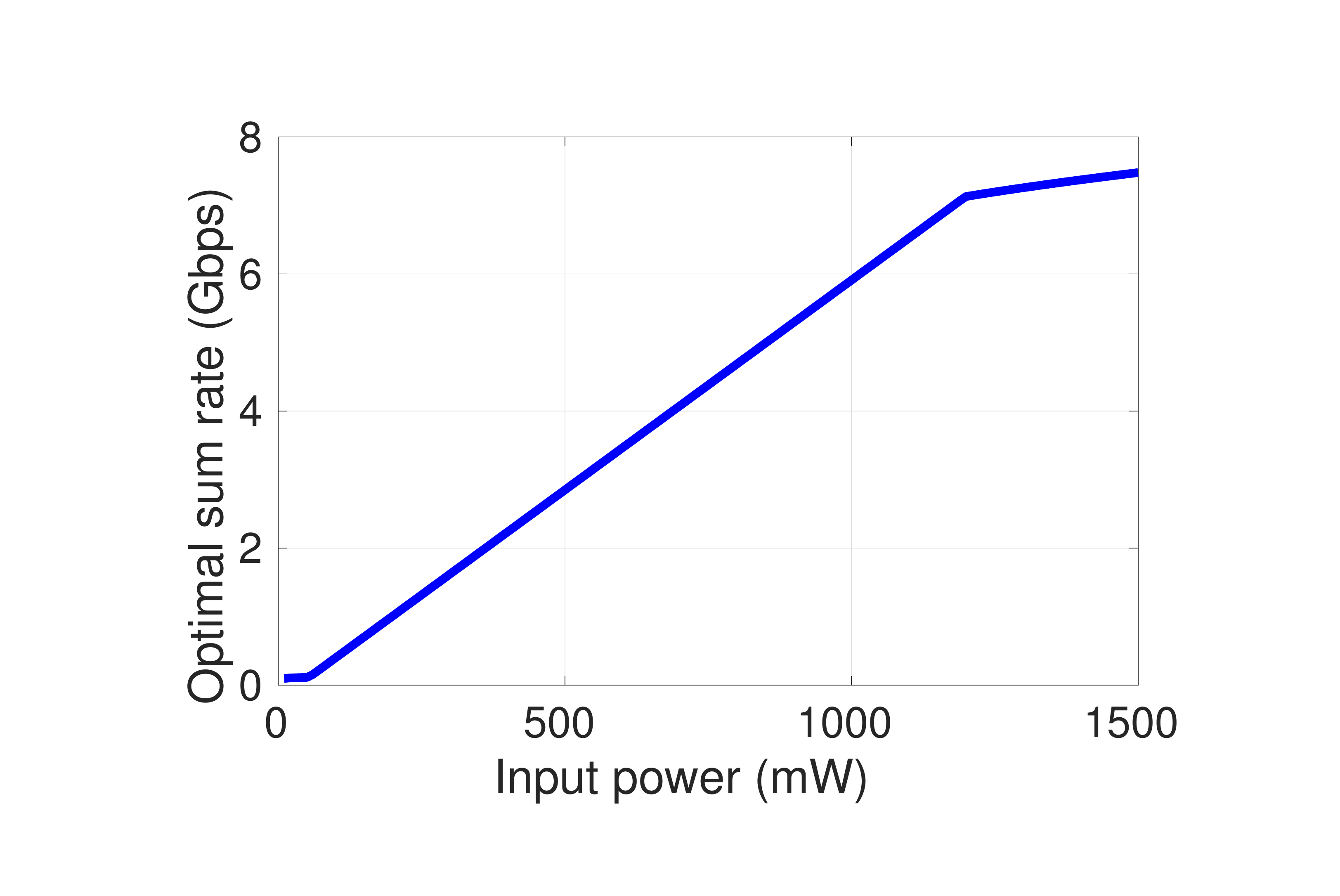}  
    }
  \end{center}
  \vspace{-0.15in}
  \caption{Optimal solution as a function of input power budget for $a=10^{-9}$}
  \vspace{-0.2in}
  \label{HB_effect}
\end{figure}

In order to investigate behavior of the optimal power and bandwidth allocation as a function of the physical characteristic of ADC (i.e., $a=c_{ox}2^{r_{ADC}}$), we note that 
 $W_\text{m}^*$ decreases as the power consumption by ADC component increases. It is straightforward  to show this by
taking the derivative of $W_\text{m}^*$ with respect to $a$ and see that
$
\frac{\partial W_\text{m}^*}{\partial a} < 0.
$
Therefore, the optimal bandwidth allocated to the mmWave interface decreases as the power consumption by ADC increases.  
Figure \ref{fig:mmWave_bandwidh_ADC} shows the optimal mmWave bandwidth allocation as a function of the scaling factor $a$, noting that a larger $a$ indicates that the ADC consumes more power for a fixed bandwidth. Moreover, Fig. \ref{fig:sum_rate_ADC} demonstrates the achievable sum rate as a function of $a$. From the results, we observe that the optimal bandwidth allocation and the achievable sum rate decreases as $a$ increases. In Fig. \ref{fig:mmWave_bandwidth_power} and \ref{fig:mmWave_sum_rate}, we investigate behavior of the optimal bandwidth allocation and achievable sum rate as a function of the power budget $P_{\max}$. From the results, we observe that the optimal allocated bandwidth to the mmWave interface and the achievable sum rate increases as the input power increases, as expected.

\textbf{(II) Full mmWave Bandwidth Allocation ($\boldsymbol{\mu2 > 0}$):} Similar to Case 1, if due to the channel conditions, the inequality $\mu_2 > \mu_1$ holds, 
then from the complimentary slackness conditions, we conclude that the mmWave bandwidth should be fully allocated, i.e., $W_\text{m}^* = W_\text{m}^\text{max}$. In this case, the optimal solution is obtained as follows:  
\begin{equation}
P_\text{m}^* = \max\left(0, \frac{aW_\text{m}^\text{max}}{D}\right),
\end{equation} 
\begin{equation}
P_\text{sub-6}^* = \max\left(0, \frac{P_{\max}-EW_\text{m}^{\max}}{D+1}\right), 
\end{equation}
\begin{equation}
W_\text{sub-6}^* = \max\left(0, \frac{P_\text{max} - E W_\text{m}^{\max}}{n_r E}\right),
\end{equation}
where 
$$
D = \omega \left(\frac{\sum_{i=1}^n \log(\lambda_i^2)}{n} - \log\frac{n_t}{a} - 1 \right)\ \text{and} \ \ E=a+\frac{a}{D}.
$$
Similar to Case 1, we note that only one of the interfaces (i.e., mmWave) utilizes its full bandwidth, and the bandwidth allocated to the other interface (i.e., sub-6 GHz) decreases as the energy consumption by the ADC component increases.  

\textbf{(III) Special case of negligible ADC consumption in sub-6 GHz:} In the case that ADC power consumption in the sub-6 GHz interface is negligible compared with the mmWave interface, it is optimal to always allocate full bandwidth to the sub-6 GHz interface. In particular, the power constraint \eqref{power-constraint} is simplified to:
$$
P_{\text{sub-6}} \large{\mathds{1}}_{W_\text{sub-6} > 0} +P_{\text{m}} \large{\mathds{1}}_{W_\text{m} > 0}+ a W_{\text{m}}  \leq P_{\max}.
$$ 
In this case, the optimal solution always falls back to Case 1 discussed earlier where we have $\mu_1 > 0$, resulting in  allocating full bandwidth to the sub-6 GHz interface, as expected. Moreover, it should be noted if the power consumption by the mmWave ADC is not taken into account, then the optimal solution allocates full bandwidth to both interfaces, as it has been prevalent in the previous works.

\textbf{(IV) Low SNR regime:} In order to derive the optimal values in \eqref{eq:power-optimal}--\eqref{eq:BW-optimal}, we assumed that the system operates at high SNR regime. We can extend the previous results into the low SNR setting as well.  We apply the approximation $\log (1+x) \approx x \log_2 e$ for $x$ small, and obtain the following approximations for the sub-6 GHz and mmWave achievable rates at low SNR regime: 
$$
\mathcal{R}_\text{sub-6} \approx \sum_{i=1}^n \frac{P_\text{sub-6}}{n_t}\lambda_i^2 \log_2 e,
\quad \text{and} \quad 
\mathcal{R}_\text{m} \approx P_\text{m} \left|\vc{w}_r^H \mathbf{H}_{\text{m}} \vc{w}_t\right|^2 \log_2 e. 
$$
As a result, the KKT conditions are obtained as follows:
\begin{equation}
  \left\{
  \begin{array}{l l}
 n_r a \mu_0 + \mu_1 & =0 ,\\ \vspace{.1cm}
    
 \frac{\log_2 e}{n_t} \sum_{i=1}^{n}\lambda_i^2  & = \mu_0 - \mu_3, \\ \vspace{.1cm}  
  
a  \mu_0 + \mu_2 & =0, \vspace{.14cm} \\  
    
\left|\vc{w}_r^H \mathbf{H}_{\text{m}} \vc{w}_t\right|^2 \log_2 e& = \mu_0 - \mu_4.
  \end{array} \right.
  \label{kkt-low-SNR}
\end{equation}
Similar to the high SNR setting, depending on the channel conditions, we consider different scenarios under which only one of the interfaces becomes active. In the first scenario, we assume that the mmWave channel has a better channel condition than the sub-6 GHz channel. Therefore, from \eqref{kkt-low-SNR} we have:
$
\mu_0 - \mu_3 < \mu_0 - \mu_4   \Rightarrow 0 \leq  \mu_4 < \mu_3. 
$
However, from the complementary slackness condition, we have $P_\text{sub-6} \mu_3 = 0$. Therefore, the optimal allocated power to the sub-6 GHz interface should be zero since $0<\mu_3$. This is in agreement with intuition that when the system operates at low SNR regime, the input power is allocated to the interface that has a better channel condition. A similar argument holds when the sub-6 GHz channel has a better channel condition compared with the mmWave interface, which results in zero power allocation to the mmWave interface and transmission across the sub-6 GHz interface only.  

\subsection{Partial Channel State Information at the Transmitter}
\label{sec:partial-CSIT}
In MIMO systems, calculating the achievable rate boils down to finding the transmit covariance
matrices that includes finding the
optimum transmit directions and the optimum power
allocation across the transmit antenna array.  When the channel is changing over time due to fading, and perfect and instantaneous CSI is known both at the receiver and at the transmitter side, the optimal power allocation across the antenna elements extend to water-filling over both time and space (i.e., the eigenmodes). However, in most of the MIMO wireless communications scenarios, it is unrealistic
to assume that the transmitter side has the perfect knowledge
of the instantaneous CSI. In such cases, it might be more
realistic to assume that only the receiver side can perfectly
estimate the instantaneous CSI, while the transmitter side has
only a partial knowledge of the channel. 

Thus far, we studied the optimal power and bandwidth allocation assuming that the sub-6 GHz transmitter does not have channel state information. In this section, we extend the previous results to the case in which there is partial CSIT.
In order to optimally allocate power across the sub-6 GHz and mmwave interfaces, we require that the instantaneous power constraint to be satisfied. Therefore, the optimal covariance matrix $\mathbf{K}_{\mathbf{xx}}$ should satisfy the power constraint at all times, independent of the channel state. Under this assumption, using the water-filling method and deriving the optimal covariance matrix becomes challenging. In order to resolve this issue, we obtain the optimal power allocation for the ``worst and best'' sub-6 GHz channel conditions. Specifically, we derive an upper and lower bound on the sub-6 GHz achievable rate and then obtain the optimal covariance matrix.

\textbf{Partial CSIT model:} Similar to the model presented in \cite{abdelaziz2017compound}, we assume that the sub-6 GHz transmitter has access to \emph{partial CSIT} such that the transmitter still does not know the \emph{exact realization of the channel}, but rather, it knows that the channel matrix belongs to a known compact  set  defined as follows: 
$$
\mathcal{S} = \left\{\mathbf{H}_\text{sub-6}: \mathbf{H}_\text{sub-6} = \bar{\mathbf{H}}_\text{sub-6} + \Delta\mathbf{H}_\text{sub-6}, \ \ |\Delta\mathbf{H}_\text{sub-6}|_2 \leq \epsilon,  \ \ \bar{\mathbf{H}}_\text{sub-6} = \lambda^{1/2} \mathbf{v}\mathbf{u}^H \right\},
$$
where $\mathbf{v} \in \mathcal{C}^{n_r\times 1}$ and $\mathbf{u} \in \mathcal{C}^{n_t\times 1}$. In this model, $\bar{\mathbf{H}}_\text{sub-6}$ is the mean channel information that corresponds to the LOS component. In general, LOS component causes the channel variations to be centered around a mean value with unit rank whose gain depends mainly on the distance between the transmitter and receiver, array configuration and respective array orientation. Hence, we assume that $\bar{\mathbf{H}}_\text{sub-6}$ is of unit rank.
On the other hand,  $\Delta\mathbf{H}_\text{sub-6}$ captures the NLOS components and represents the variations around the mean value such that $|\Delta\mathbf{H}_\text{sub-6}|_2 \leq \epsilon$. We further assume that the channel error matrix $
\Delta \mathbf{H}_\text{\text{sub-6}}$ has zero mean
and i.i.d. Gaussian entries, each with variance of $\sigma_e ^2$.
Therefore, we have the following mapping between the channel matrices: 
$$
\bar{\mathbf{H}}_\text{sub-6} \Longleftrightarrow \mathbf{H}_\text{sub-6}^{\text{LOS}} \ \ \text{and} \ \ \Delta\mathbf{H}_\text{sub-6} \Longleftrightarrow \mathbf{H}_\text{sub-6}^{\text{NLOS}}. 
$$

\textbf{(I) Lower Bound on the sub-6 GHz Achievable Rate:} Given that the transmitter has partial CSI according to the presented model, we obtain a lower bound on the achievable rate of the sub-6 GHz interface, and derive the corresponding optimal covariance matrix. In order to derive a lower bound on the sub-6 GHz achievable rate, we consider the \emph{worst sub-6 GHz channel condition} that renders a lower bound on the sub-6 GHz achievable rate. Then, we optimize the physical layer resource allocation for the worst sub-6 GHz channel. From \cite{abdelaziz2017compound}, we adopt the following result. 

\begin{theorem} 
Assuming the presented CSIT model, and for any nonnegative definite covariance matrix $\mathbf{K}_\mathbf{xx}$, if the achievable data rate of the sub-6 GHz interface is denoted by $\mathcal{R}_\text{sub-6}(\mathbf{H}_\text{sub-6}, \mathbf{K}_\mathbf{xx})$, then we have:
\begin{equation}
\mathcal{R}_\text{sub-6}(\mathbf{H}_\text{sub-6}, \mathbf{K}_\mathbf{xx}) \geq \mathcal{R}_\text{sub-6}({\mathbf{H}}^*_\text{sub-6}, \mathbf{K}_\mathbf{xx}),
\end{equation}
in which ${\mathbf{H}}^*_\text{sub-6}$ is the \textbf{worst sub-6 GHz channel} such that ${\mathbf{H}}_\text{sub-6}^{*H}{\mathbf{H}}^*_\text{sub-6} = [(\lambda^{1/2} - \epsilon)^+]^2\mathbf{uu}^H$ where $(x)^+ = \max(0,x)$.  
\label{theorem-lower-bound}
\end{theorem}

\begin{proof}
\small
\begin{align}
\mathcal{R}_\text{sub-6}(\mathbf{H}_\text{sub-6}, \mathbf{K}_\mathbf{xx}) =  W_{\text{sub-6}} \log \det \left( \mathbf{I}+ \frac{\mathbf{H}_{\text{sub-6}}^H \mathbf{K}_\mathbf{xx} \mathbf{H}_{\text{sub-6}}}{W_{\text{sub-6}}}  \right) \stackrel{(a)}{=} \sum_{i=1}^{n_t} W_\text{sub-6} \log\left(1+\frac{1}{W_\text{sub-6}}\sigma_i\left(\mathbf{H}_\text{sub-6}^H \mathbf{K}_\mathbf{xx} \mathbf{H}_\text{sub-6}\right)\right) \nonumber &\\ \stackrel{(b)}{=} \sum_{i=1}^{n_t} W_\text{sub-6} \log\left(1+\frac{1}{W_\text{sub-6}}\lambda_i\left((\bar{\mathbf{H}}_\text{sub-6}+\Delta\mathbf{H}_\text{sub-6})\mathbf{K}_\mathbf{xx}^{1/2}\right)^2\right)  \stackrel{(c)}{\geq} W_\text{sub-6} \log\left(1+\frac{1}{W_\text{sub-6}}[(\lambda^{1/2}-\epsilon)^+]^2\lambda_1(\mathbf{K}_\mathbf{xx}) \right) &\nonumber\\ = \mathcal{R}_\text{sub-6}({\mathbf{H}}^*_\text{sub-6}, \mathbf{K}_\mathbf{xx})&,
\label{eq:partial-csi}
\end{align}
\normalsize
where $\sigma_i (\mathbf{X})$ and $\lambda_i (\mathbf{X})$ denote the $i$-th eigenvalue and singular value of matrix $\mathbf{X}$, respectively. In this case, (a) follows from the determinant properties that the determinant of a matrix is equal to the product of its eigenvalues. Moreover, (b) follows from the fact that:
$$\sigma_i\left(\mathbf{H}_\text{sub-6}^H\mathbf{K}_\mathbf{xx} \mathbf{H}_\text{sub-6}\right)~=~\lambda_i\left(\left(\bar{\mathbf{H}}_\text{sub-6}+\Delta\mathbf{H}_\text{sub-6}\right)\mathbf{K}_\mathbf{xx}^{1/2}\right)^2.
$$ 
Finally,  (c) is concluded  from the following singular value inequality \cite{schaefer2015secrecy}, and the fact that $\lambda_i(\bar{\mathbf{H}}_\text{sub-6}) = 0 \ \forall i > 1$.

\begin{lemma}[Singular value inequality]
Let $\mathbf{A}, \mathbf{B}$ and $\mathbf{C}$ be $n \times m$ and $m \times m$ matrices, and $\lambda_i (\mathbf{X})$ denotes the $i$-th singular value of matrix $\mathbf{X}$ such that $\{\lambda_i\}$ is in decreasing order. Then,
\begin{equation}
\lambda_i \left( (\mathbf{A}+\mathbf{B})\mathbf{C} \right) \geq \left(\lambda_i(\mathbf{A}) - \lambda_1 (\mathbf{B}) \right)^+ \lambda_i(\mathbf{C})
\end{equation} 
\end{lemma}
Proof is provided in \cite{schaefer2015secrecy}. 
\end{proof}

\emph{Theorem \ref{theorem-lower-bound} implies that the optimal input covariance matrix $\mathbf{K}_{\mathbf{xx}}$ for the worst sub-6 GHz channel has to be of unit rank.} That is because the matrix ${\mathbf{H}}^{*H}{\mathbf{H}}^*$ is shown to be of unit rank, and thus it has only one eigenvalue equals to $P_\text{sub-6}$. Its corresponding eigenvector $\mathbf{q}$ with $||\mathbf{q} ||=1$ satisfies the sub-6 GHz interface transmit power constraint. Therefore, the optimal covariance matrix $\mathbf{K}_\mathbf{xx}$ is written as follows:
\begin{equation}
\mathbf{K}_\mathbf{xx} = P_\text{sub-6} \mathbf{q q}^H \ \ \text{s.t.} \ \ ||\mathbf{q} ||=1,
\end{equation}
where the optimal $\mathbf{q}$ can be obtained by the null steering solution \cite{friedlander1989performance}.

\emph{Compared with \eqref{eq:optimal-covariance-no-csi}, we note that there is a power gain of factor $n_t$ when the transmitter can track the mean channel information.} Without channel knowledge at the transmitter, the transmit energy is spread out equally across all directions in $\mathcal{C}^{n_t}$. With transmitter knowledge of the mean channel information (i.e., LOS direction), the energy can now be focused on only one direction. However, it should be noted that the power gain accounts for the lower bound on the achievable rate. From the result on the optimal sub-6 GHz covariance matrix for the worst sub-6 GHz channel, we can apply the same method as in the scenario with no CSIT, and obtain the optimal power and bandwidth allocation across the interfaces. 

\textbf{ (II) Upper Bound on the sub-6 GHz Achievable Rate:} In this section, we derive an upper bound on the sub-6 GHz achievable rate and express the optimal covariance matrix that achieves the upper bound. Similar to the previous section, the physical layer resources can then be optimally allocated across the  sub-6 GHz and mmWave interfaces. We assume that the model of partial CSIT is valid such that the transmitter CSI includes a known component $\bar{\mathbf{H}}_\text{sub-6}$ that corresponds to the LOS component, and an error term $\Delta\mathbf{H}_\text{sub-6}$ that corresponds to the NLOS component.

For the purpose of our optimal resource allocation across the sub-6 GHz and mmWave interfaces, we require to obtain the optimal covariance matrix that is expressed in terms of total sub-6 GHz transmit power $P_\text{sub-6}$. In this case, we can show that:
$$
\mathcal{R}_\text{sub-6}(\mathbf{H}_\text{sub-6}, \mathbf{K}_\mathbf{xx}) \leq W_{\text{sub-6}} {\log \det \left( \mathbf{I}+ \frac{1}{W_{\text{sub-6}}} \mathbf{K}_\mathbf{xx}(\bar{\mathbf{H}}_{\text{sub-6}}^H \bar{\mathbf{H}}_{\text{sub-6}} + \sigma_e^2 \mathbf{I})\right)},
$$
As a result, the transmitter observes the equivalent channel covariance matrix $\boldsymbol{\Sigma} := \bar{\mathbf{H}}_{\text{sub-6}}^H \bar{\mathbf{H}}_{\text{sub-6}} + \sigma_e^2 \mathbf{I}$, and thus eigenvectors of the optimal covariance matrix $\mathbf{K}_\mathbf{xx}$ are equal to the eigenvectors of matrix $\boldsymbol{\Sigma}$. Given the knowledge of the channel covariance matrix, the authors in \cite{soysal2007optimum} derived the optimal power allocation across the transmit antennas such that the power $p_i$ is allocated to the $i$-th transmit antenna:
\begin{equation}
p_i = \frac{p_i E_i (\mathbf{p})}{\sum_{j=1}^{n_t} p_j E_j (\mathbf{p})} P_\text{sub-6}, 
\label{eq:optimal-power-RF}
\end{equation}
where $\mathbf{p} = (p_1, ..., p_{n_t})$ is the vector of optimal power allocated across the $n_t$ sub-6 GHz transmit antennas, and $E_i(\mathbf{p})$ is expressed in terms of eigenvalues of the channel covariance matrix and the power allocation vector $\mathbf{p}$, i.e.,:
$$
E_i(\mathbf{p}) = E \left[ \frac{\sigma_i({\boldsymbol{\Sigma}}) \mathbf{z}_i ^ H \mathbf{A}_i ^ {-1} \mathbf{z}_i}{1+p_i \sigma_i({\boldsymbol{\Sigma}}) \mathbf{z}_i ^ H \mathbf{A}_i ^ {-1} \mathbf{z}_i}\right],
$$
where $\mathbf{A}_i = \mathbf{I} + \sum_{j=1}^{n_t} p_j \sigma_j({\boldsymbol{\Sigma}}) \mathbf{z}_j \mathbf{z}_j^H - p_i \sigma_i({\boldsymbol{\Sigma}}) \mathbf{z}_i \mathbf{z}_i^H$, and $\mathbf{z}_i$ is the $i$-th column of $\mathbf{Z}$ that is used to convert the channel covariance matrix $\boldsymbol{\Sigma}$ to the channel matrix $\bar{\mathbf{H}}_\text{sub-6}$ by $\bar{\mathbf{H}}_\text{sub-6} = \mathbf{Z} \boldsymbol{\Sigma}^{1/2}$. Entries of $\mathbf{Z}$ are i.i.d., zero-mean, unit-variance complex Gaussian random variables \cite{soysal2007optimum}. Finally, from \eqref{eq:optimal-power-RF}, the optimal power allocation to the sub-6 GHz transmit antennas can be obtained by a fixed point algorithm \cite{soysal2007optimum}. Therefore, the optimal covariance matrix $\mathbf{K}_\mathbf{xx}$ can be expressed in terms of the total sub-6 GHz transmit power $P_\text{sub-6}$, and similar to the previous section, we derive the optimal power and bandwidth allocation across the sub-6 GHz and mmWave interfaces.






\section{Energy Efficiency Maximization}
\label{sec:problem2-solution}
Problem 2 formulated in \eqref{eq:problem2} through \eqref{eq:problem2-power-constraint} is aimed to optimize the normalized sum rate in terms of bit/sec/watt or bit/joule.  Our observations on the solution of this problem for the SISO case indicate that (i) if the transmission power is much lower than component consumption in any interface, the objective function is decreasing in bandwidth for that interface. This suggests that the optimal solution is not to allocate any bandwidth for that interface; (ii) on the flip side, if the transmission power is much higher than the component consumption, then the objective function is increasing in bandwidth for that platform. This suggests that the optimal solution is to allocate full bandwidth for that platform; (iii) otherwise, the objective function is non-monotonic, increasing up to a certain bandwidth at which point it starts to decrease. Indeed, in Fig.~\ref{fig:shannon_with_components}, we observe such a case for the mmWave interface.

 The problem of normalized sum rate maximization can be viewed as an instance of energy efficiency (EE) problem defined as the ratio between the amount of data transmitted and the corresponding incurred cost in terms of power. In order to solve this problem, we use fractional programming that solves the problem in the following general form: 
\begin{align}
\max_{\mathbf{x}}  \ \frac{f(\mathbf{x})}{g(\mathbf{x})}; \\
\text{s.t.} \ \ \mathbf{x} \in \mathcal{X},
\end{align}
with $f: \mathcal{C} \subseteq \mathds{R}^n \rightarrow \mathds{R}$,  $g: \mathcal{C} \subseteq \mathds{R}^n \rightarrow \mathds{R}_+$, and $\mathcal{X} \subseteq \mathcal{C} \subseteq \mathds{R}^n$. In general, obtaining optimal solution for the fractional problems is challenging since the objective function is not concave, and thus the
strong results of convex optimization theory (i.e., KKT conditions) do not apply. However, under some assumptions on $f$ and $g$, the objective falls into the class of generalized concave functions, which makes it feasible to obtain the global solution. Considering the formulation in \eqref{eq:problem2}, since logarithm of identity matrix plus an Hermitian positive-semidefinite matrix is a matrix-concave function, the numerator is concave. Denumerator is a linear function of the sub-6 GHz and mmWave power and bandwidth variables. Therefore, the objective function is pseudo-concave. As a result, each stationary point of the objective function is a global maximizer, and KKT conditions are necessary and sufficient for optimality \cite{zappone2015energy}.

In order to solve pseudo-concave problems, some solution approaches have been proposed. Dinkelbach’s algorithm proposed in \cite{dinkelbach1967nonlinear,jagannathan1966some} is a parametric algorithm that instead of solving the original problem, solves a sequence of easier problems that converge to the global solution. Specifically, let us assume that the set of feasible solutions of Problem 2 is denoted by $\mathcal{S}$. Then, the main result of the  Dinkelbach’s algorithm is the relation between the original problem and the following function:
\begin{equation}
F(\beta) = \max_{\mathbf{x} \in \mathcal{X}} \left\{f(\mathbf{x}) - \beta g(\mathbf{x}) \right\}.
\end{equation}
In this case, $F$ exists and is continuous provided that $f$ and $g$ are continuous and $\mathcal{X}$ is compact. In addition, it is possible to show that $F$ is convex on $\mathds{R}$, and is strictly monotone decreasing on $\mathds{R}$ \cite{zappone2015energy}. The connection between $F(\beta)$ and the original problem is as follows. Consider $\mathbf{x}^* \in \mathcal{X}$ and $\beta^* = \frac{f(\mathbf{x}^*)}{g(\mathbf{x}^*)}$, then $\mathbf{x}^*$ is the optimal solution for the original problem if and only if:
\begin{equation}
\mathbf{x}^* = \argmax_{\mathbf{x} \in \mathcal{X}} \{f(\mathbf{x}) - \beta^* g(\mathbf{x}) \}.
\label{eq:dinkelbach} 
\end{equation}
Therefore, it can be shown that solving the fractional programming is equivalent to finding the roots of $F(\beta)$. Thus, Dinkelbach’s algorithm achieves this goal as provided in Algorithm 1.  
Optimality and convergence results of Dinkelbach's algorithm are shown in \cite{zappone2015energy}. 

\begin{algorithm}[t]
\caption{Dinkelbach’s Algorithm}
\begin{varwidth}{\dimexpr\linewidth-2\fboxsep-2\fboxrule\relax}
\begin{algorithmic}[1]
\State $\delta > 0$, $n=0$, $\beta_n =0$
\While{$F(\beta_n) > \delta$} 
\State $\mathbf{x}^*_n = \argmax_{\mathbf{x} \in \mathcal{X}} \{f(\mathbf{x}) - \beta_n g(\mathbf{x}) \}$
\State $F(\beta_n) = \max_{\mathbf{x} \in \mathcal{X}} \left\{f(\mathbf{x^*}) - \beta_n g(\mathbf{x^*}) \right\}.$
\State $\beta_{n+1} = \frac{f(\mathbf{x}^*_n)}{g(\mathbf{x}_n^*)}$
\State $n=n+1$
\EndWhile
\end{algorithmic}
\end{varwidth}%
\end{algorithm}

By considering the energy efficiency problem as a pseudo-concave fractional problem and applying Dinkelbach’s method, we in fact reach to the same formulation as Problem 1. In particular, Eq. \eqref{eq:dinkelbach} resembles the same equation as the Lagrange equation in \eqref{lagrange}. Therefore, the solution space will be similar, while Dinkelbach's algorithm iteratively solves the Lagrange equation instead of finding the solution in one-shot. Similar to Problem 1, various CSIT scenarios can be considered in order to first simplify the objective value in EE problem.

\section{Numerical Results}
In this section, we numerically investigate the performance of our proposed resource allocation scheme. For the mmWave technology, the carrier frequency is $30$ GHz and the \emph{total available bandwidth} is $1$ GHz. The number of transmit and receive antennas are $64$ and $16$, respectively. For the sub-6 GHz interface, the carrier frequency is $3$ GHz, the \emph{total available bandwidth} is $1$ MHz, and the number of antennas is the same as in mmWave. MmWave and sub-6 GHz channel matrices are extracted from the experimental data in \cite{mezzavilla20155g}. In the simulation results, we obtain the performance as a function of the total available power and the parameter $a = c_\text{ox} 2^{r_{\text{ADC}}}$. The former dictates the power constraint, while the latter   determines the power consumption of ADC components. 

In the first example, we consider an extreme scenario in which the ADC power consumption is very high. For this purpose, we set the scaling factor of ADC power consumption to be $a=10^{-7}$. Table \ref{tab:sum-rate} demonstrates the performance of our proposed optimal solution compared with the full bandwidth allocation. From the results, we observe that in the case of utilizing the whole available spectrum, the power consumption by ADC components exceeds the total available power. Thus, the transmit power becomes zero, and the achievable rate drops to zero as well. On the other hand, if only about $4$ MHz of the mmWave bandwidth is utilized, then the total sum rate of $2.55$ Kbps is achievable.   

\begin{table}[h]\centering
\renewcommand{\arraystretch}{1.3}
\ra{1.3}
\caption{Achievable sum rate by: (i) Full sub-6 GHz and mmWave bandwidth allocation, and (ii) Optimal allocation.}
\begin{threeparttable}
\begin{tabular}{@{}p{4cm}cc@{}}\toprule
\textbf{Resource} &  \textbf{Full Bandwidth Allocation} & \textbf{Optimal Bandwidth Allocation} \\ \toprule
Sub-6 GHz Bandwidth     &   $1$ MHz    & $1$ MHz \\ \midrule
mmWave Bandwidth   & $1$ GHz  &  $4.0597$ MHz \\ \midrule
Sub-6 GHz \emph{Transmission} Power   & $0$\tnote{1}  &  
$394$ mW \\ \midrule
mmWave \emph{Transmission} Power   & $0$\tnote{1}  &  $100$ mW \\ \midrule
\textbf{Achievable Sum Rate}   & \textbf{0}  &  \textbf{$\boldsymbol{2.55\times 10^3}$} \\ 
 \bottomrule
\end{tabular}
 \begin{tablenotes}
       \item[1] Due to high power consumption of the mmWave component, \emph{transmit power} becomes zero. 
     \end{tablenotes}
 \end{threeparttable}
  \label{tab:sum-rate}
\end{table}


%
%

\begin{figure}[h]
\centering
\includegraphics[scale=.3]{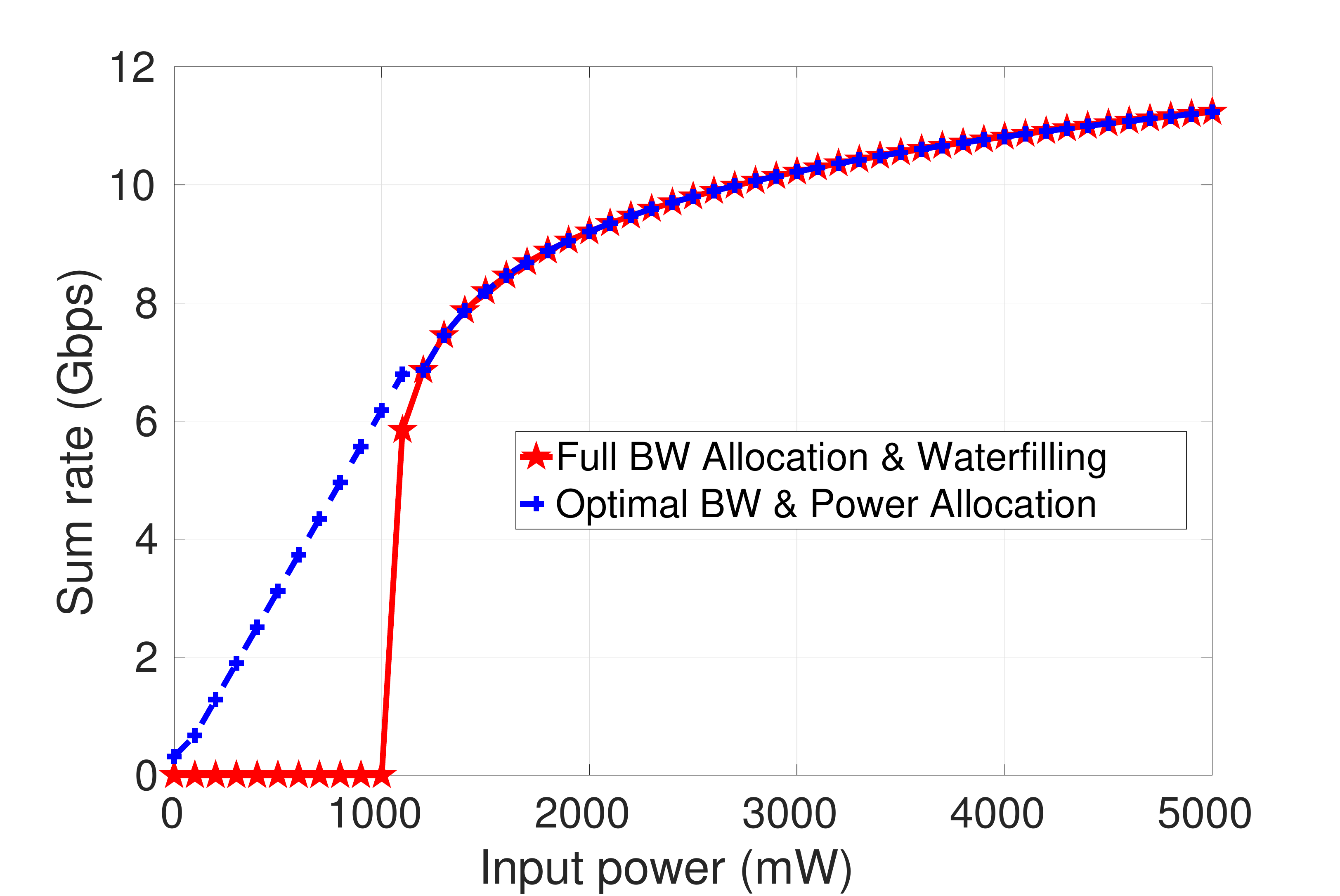}
\caption{Sum rate comparison between the optimal scheme and the full bandwidth allocation paired with the waterfilling power allocation across the sub-6 GHz and mmWave interfaces.}
\label{fig:optimal_sum_rate_vs_power}
\end{figure}

In the second scenario, we investigate the optimal solution of Problem 1 with no CSIT and as a function of available power (i.e., input power) and the power consumption of ADC components.  The results are shown in Fig. \ref{fig:optimal_sum_rate_vs_power} and Fig. \ref{fig:optimal_sum_rate_vs_ADC}, respectively. \emph{From the results, we observe that under low power scenario, it is optimal to partially use the available bandwidth to achieve a higher sum rate compared with the full bandwidth utilization. Moreover, when the ADC energy consumption increases, it is more energy efficient to partially use the available bandwidth.}

\begin{figure}
\centering
\includegraphics[scale=.3]{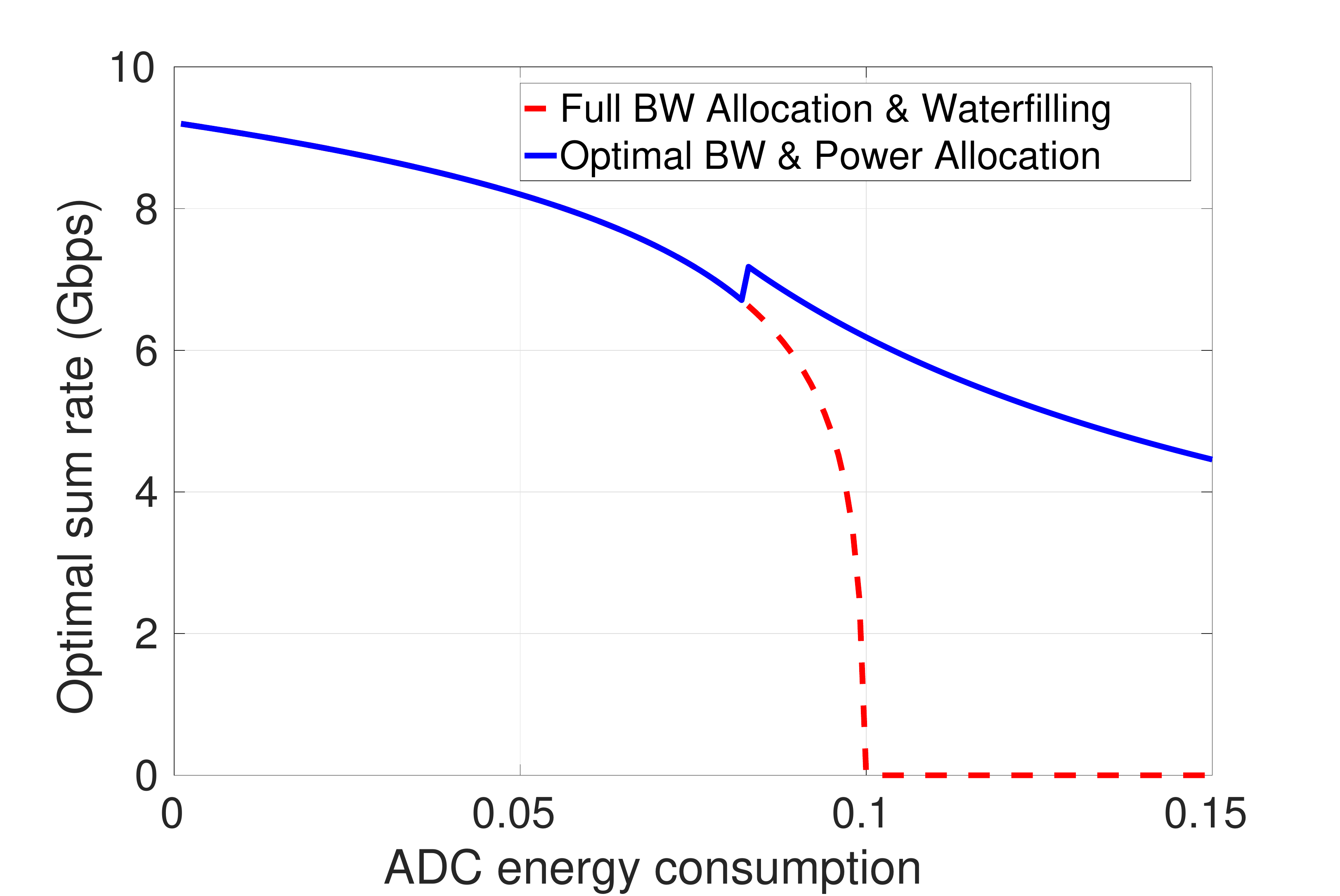}
\caption{Sum rate comparison between the optimal scheme and the full bandwidth allocation paired with the waterfilling power allocation across the sub-6 GHz and mmWave interfaces. In this figure, ADC energy consumption refers to the constant $a\times 10^8$. }
\label{fig:optimal_sum_rate_vs_ADC}
\end{figure}

Figure \ref{fig:problem-2} numerically demonstrates the optimal solution of energy efficient resource allocation. From the results, we observe that after some threshold on the mmWave bandwidth, the energy efficiency in terms of bits/sec/watt or bits/joule decreases for the mmWave interface. However, under our optimal resource allocation scheme, the increasing trend of the energy efficiency as a function of bandwidth is preserved due to partially utilizing the bandwidth.

\begin{figure}
\centering
\includegraphics[scale=.3]{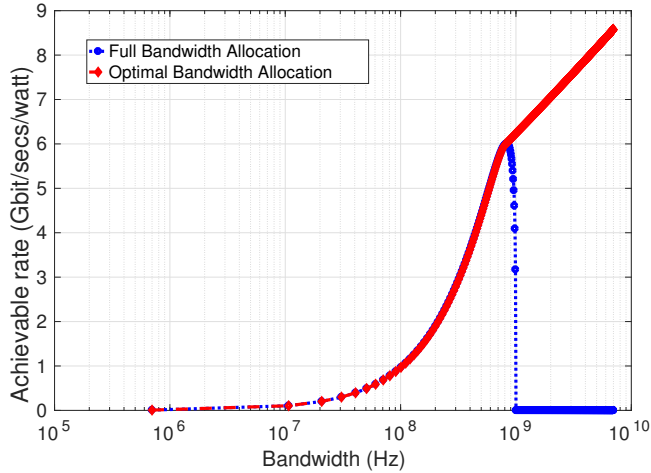}
\caption{Optimal solution of Problem 2 vs. full bandwidth allocation}
\label{fig:problem-2}
\end{figure}

\section{Conclusion}
\label{conclusion}
In this paper, we considered an integrated sub-6 GHz/mmWave architecture and proposed a joint power and bandwidth allocation framework in order to maximize the achieving sum rate and energy efficiency in the sub-6 GHz/mmWave system. In order to maximize the achievable sum rate under the transmitter and receiver power constraints, our formulation explicitly takes into account the energy consumption in integrated-circuit components. In addition, we investigated the optimal solution under various conditions such as when there is no channel state information at the transmitter or when there is only partial information available. From our optimal results, \emph{we observe that despite the availability of huge bandwidths at the mmWave interface, under some circumstances (e.g., low input power or ADC components with high energy consumption), it is more efficient to partially utilize the mmWave bandwidth.} We emphasize that physical layer resource allocation is of great importance in mmWave systems since the energy consumption by components increases with bandwidth, and thus they can incur large overhead in terms of power consumption. Hence, the power and bandwidth should be optimally allocated in order to avoid the heavy burden of components power consumption.

\bibliographystyle{ieeetr}
\bibliography{../References}{}

\begin{thebibliography}{10}

\bibitem{khan2011mmwave}
F.~Khan and Z.~Pi, ``mm{W}ave mobile broadband ({MMB}): Unleashing the
  3--300{GH}z spectrum,'' in {\em 34th IEEE Sarnoff Symposium}, 2011.

\bibitem{rappaport2013millimeter}
T.~S. Rappaport, S.~Sun, R.~Mayzus, H.~Zhao, Y.~Azar, K.~Wang, G.~N. Wong,
  J.~K. Schulz, M.~Samimi, and F.~Gutierrez, ``Millimeter wave mobile
  communications for 5{G} cellular: It will work!,'' {\em Access, IEEE},
  vol.~1, pp.~335--349, 2013.

\bibitem{zhao201328}
H.~Zhao, R.~Mayzus, S.~Sun, M.~Samimi, J.~K. Schulz, Y.~Azar, K.~Wang, G.~N.
  Wong, F.~Gutierrez, and T.~S. Rappaport, ``28 ghz millimeter wave cellular
  communication measurements for reflection and penetration loss in and around
  buildings in new york city,'' in {\em 2013 IEEE International Conference on
  Communications (ICC)}, pp.~5163--5167, IEEE, 2013.

\bibitem{lu2012modeling}
J.~S. Lu, D.~Steinbach, P.~Cabrol, and P.~Pietraski, ``Modeling human blockers
  in millimeter wave radio links,'' {\em ZTE Communications}, vol.~10, no.~4,
  pp.~23--28, 2012.

\bibitem{orhan2015low}
O.~Orhan, E.~Erkip, and S.~Rangan, ``Low power analog-to-digital conversion in
  millimeter wave systems: Impact of resolution and bandwidth on performance,''
  in {\em Information Theory and Applications Workshop (ITA)}, pp.~191--198,
  IEEE, 2015.

\bibitem{phan2010reconfigurable}
A.~T. Phan and R.~Farrell, ``Reconfigurable multiband multimode {LNA} for
  {LTE/GSM, WiMAX, and IEEE 802.11.} a/b/g/n,'' in {\em Electronics, Circuits,
  and Systems (ICECS), 2010 17th IEEE International Conference on}, pp.~78--81,
  IEEE, 2010.

\bibitem{liang20070}
K.-H. Liang, H.-Y. Chang, and Y.-J. Chan, ``A 0.5--7.5 {GHz} ultra low-voltage
  low-power mixer using bulk-injection method by 0.18-$\mu$m cmos technology,''
  {\em IEEE Microwave and Wireless Components Letters}, vol.~17, no.~7,
  pp.~531--533, 2007.

\bibitem{rappaport2011state}
B.~T.~S. Rappaport, J.~N. Murdock, and F.~Gutierrez, ``State of the art in
  60-ghz integrated circuits and systems for wireless communications,'' {\em
  Proceedings of the IEEE}, vol.~99, no.~8, pp.~1390--1436, 2011.

\bibitem{hashemi2017hybrid}
M.~Hashemi, C.~E. Koksal, and N.~B. Shroff, ``Dual sub-6 {GHz} -- millimeter
  wave beamforming and communications to achieve low latency and high energy
  efficiency in {5G} systems,'' {\em arXiv preprint arXiv:1701.06241}, 2017.

\bibitem{singh2009limits}
J.~Singh, O.~Dabeer, and U.~Madhow, ``On the limits of communication with
  low-precision analog-to-digital conversion at the receiver,'' {\em IEEE
  Transactions on Communications}, vol.~57, no.~12, 2009.

\bibitem{mo2014high}
J.~Mo and R.~W. Heath, ``High {SNR} capacity of millimeter wave {MIMO} systems
  with one-bit quantization,'' in {\em Information Theory and Applications
  Workshop (ITA)}, pp.~1--5, IEEE, 2014.

\bibitem{mezghani2007ultra}
A.~Mezghani and J.~A. Nossek, ``On ultra-wideband {MIMO} systems with 1-bit
  quantized outputs: Performance analysis and input optimization,'' in {\em
  International Symposium on Information Theory (ISIT)}, pp.~1286--1289, IEEE,
  2007.

\bibitem{bai2013optimization}
Q.~Bai, A.~Mezghani, and J.~A. Nossek, ``On the optimization of {ADC}
  resolution in multi-antenna systems,'' in {\em Wireless Communication Systems
  (ISWCS 2013), Proceedings of the Tenth International Symposium on}, pp.~1--5,
  VDE, 2013.

\bibitem{el2014spatially}
O.~El~Ayach, S.~Rajagopal, S.~Abu-Surra, Z.~Pi, and R.~W. Heath, ``Spatially
  sparse precoding in millimeter wave {MIMO} systems,'' {\em Wireless
  Communications, IEEE Transactions on}, vol.~13, no.~3, pp.~1499--1513, 2014.

\bibitem{alkhateeb2014mimo}
A.~Alkhateeb, J.~Mo, N.~Gonzalez-Prelcic, and R.~W. Heath, ``{MIMO} precoding
  and combining solutions for millimeter-wave systems,'' {\em IEEE
  Communications Magazine}, vol.~52, no.~12, pp.~122--131, 2014.

\bibitem{collonge2004influence}
S.~Collonge, G.~Zaharia, and G.~E. Zein, ``Influence of the human activity on
  wide-band characteristics of the 60 {GH}z indoor radio channel,'' {\em
  Wireless Communications, IEEE Transactions on}, vol.~3, no.~6,
  pp.~2396--2406, 2004.

\bibitem{roh2014millimeter}
W.~Roh, J.-Y. Seol, J.~Park, B.~Lee, J.~Lee, Y.~Kim, J.~Cho, K.~Cheun, and
  F.~Aryanfar, ``Millimeter-wave beamforming as an enabling technology for 5{G}
  cellular communications: theoretical feasibility and prototype results,''
  {\em IEEE Communications Magazine}, vol.~52, no.~2, 2014.

\bibitem{adhikary2014joint}
A.~Adhikary, E.~Al~Safadi, M.~K. Samimi, R.~Wang, G.~Caire, T.~S. Rappaport,
  and A.~F. Molisch, ``Joint spatial division and multiplexing for mm-wave
  channels,'' {\em IEEE Journal on Selected Areas in Communications}, vol.~32,
  no.~6, pp.~1239--1255, 2014.

\bibitem{aliestimating}
A.~Ali, N.~Prelcic, and R.~Heath, ``Estimating millimeter wave channels using
  out-of-band measurements,'' {\em Information Theory and Applications Workshop
  (ITA)}, 2016.

\bibitem{nitsche2015steering}
T.~Nitsche, A.~B. Flores, E.~W. Knightly, and J.~Widmer, ``Steering with eyes
  closed: mm-wave beam steering without in-band measurement,'' in {\em Computer
  Communications (INFOCOM), IEEE Conference on}, pp.~2416--2424, IEEE, 2015.

\bibitem{ali2017millimeter}
A.~Ali, N.~Gonz{\'a}lez-Prelcic, and R.~W. Heath~Jr, ``Millimeter wave
  beam-selection using out-of-band spatial information,'' {\em arXiv preprint
  arXiv:1702.08574}, 2017.

\bibitem{tse2005fundamentals}
D.~Tse and P.~Viswanath, {\em Fundamentals of wireless communication}.
\newblock Cambridge university press, 2005.

\bibitem{boyd2004convex}
S.~Boyd and L.~Vandenberghe, {\em Convex optimization}.
\newblock Cambridge university press, 2004.

\bibitem{abdelaziz2017compound}
A.~Abdelaziz, C.~E. Koksal, H.~E. Gamal, and A.~D. Elbayoumy, ``On the compound
  {MIMO} wiretap channel with mean feedback,'' {\em arXiv preprint
  arXiv:1701.07518}, 2017.

\bibitem{schaefer2015secrecy}
R.~F. Schaefer and S.~Loyka, ``The secrecy capacity of compound gaussian {MIMO}
  wiretap channels,'' {\em IEEE Transactions on Information Theory}, vol.~61,
  no.~10, pp.~5535--5552, 2015.

\bibitem{friedlander1989performance}
B.~Friedlander and B.~Porat, ``Performance analysis of a null-steering
  algorithm based on direction-of-arrival estimation,'' {\em IEEE Transactions
  on Acoustics, Speech, and Signal Processing}, vol.~37, no.~4, pp.~461--466,
  1989.

\bibitem{soysal2007optimum}
A.~Soysal and S.~Ulukus, ``Optimum power allocation for single-user mimo and
  multi-user mimo-mac with partial csi,'' {\em IEEE Journal on Selected Areas
  in Communications}, vol.~25, no.~7, 2007.

\bibitem{zappone2015energy}
A.~Zappone, E.~Jorswieck, {\em et~al.}, ``Energy efficiency in wireless
  networks via fractional programming theory,'' {\em Foundations and
  Trends{\textregistered} in Communications and Information Theory}, vol.~11,
  no.~3-4, pp.~185--396, 2015.

\bibitem{dinkelbach1967nonlinear}
W.~Dinkelbach, ``On nonlinear fractional programming,'' {\em Management
  science}, vol.~13, no.~7, pp.~492--498, 1967.

\bibitem{jagannathan1966some}
R.~Jagannathan, ``On some properties of programming problems in parametric form
  pertaining to fractional programming,'' {\em Management Science}, vol.~12,
  no.~7, pp.~609--615, 1966.

\bibitem{mezzavilla20155g}
M.~Mezzavilla, S.~Dutta, M.~Zhang, M.~R. Akdeniz, and S.~Rangan, ``{5G} mmwave
  module for the ns-3 network simulator,'' in {\em Proceedings of the 18th ACM
  International Conference on Modeling, Analysis and Simulation of Wireless and
  Mobile Systems}, pp.~283--290, ACM, 2015.

\end{thebibliography}

\end{document}